\documentclass{article}

\usepackage{PRIMEarxiv}

\usepackage[utf8]{inputenc} 
\usepackage[T1]{fontenc}    
\usepackage{hyperref}       
\usepackage{url}            
\usepackage{booktabs}       
\usepackage{amsfonts}       
\usepackage{nicefrac}       
\usepackage{microtype}      
\usepackage{lipsum}
\usepackage{fancyhdr}       
\usepackage{graphicx}       
\graphicspath{{media/}}     

\usepackage{mathrsfs}
\usepackage{adjustbox}
\usepackage{algorithmicx}
\usepackage{algpseudocode}
\usepackage{amsbsy}
\usepackage{amsfonts}
\usepackage{amsmath}
\usepackage{amssymb}
\usepackage{amsthm}
\usepackage{array}
\usepackage{bm}
\usepackage{booktabs}
\usepackage{caption}
\usepackage{color}
\usepackage{epsfig}
\usepackage{epstopdf}
\usepackage{eurosym}
\usepackage{fancyhdr}
\usepackage{float}
\usepackage{geometry}
\usepackage{graphics}
\usepackage{graphicx}
\usepackage{indentfirst}
\usepackage{lastpage}
\usepackage{longtable}
\usepackage{lscape}
\usepackage{makecell}
\usepackage{multicol}
\usepackage{multirow}
\usepackage{rotating}
\usepackage{stfloats}
\usepackage{subfigure}
\usepackage{supertabular}
\usepackage{svg}
\usepackage{textcomp}
\usepackage{threeparttable}
\usepackage{tabularx}
\usepackage{url}
\usepackage{warpcol}
\usepackage{xcolor}
\usepackage{algorithm}
\usepackage[numbers]{natbib}
\newtheorem{theorem}{Theorem}

\newtheorem{definition}{Definition}
\newtheorem{property}{property}

\title{Privacy-preserving Security Inference Towards Cloud-Edge Collaborative Using Differential Privacy
\thanks{\textit{\underline{Corresponding author}}: 
\textbf{Xingshu Chen}} 
}

\author{
  Yulong Wang, Xingshu Chen, Qixu Wang\\
  School of Cyber Science and Engineering \\
  Sichuan University \\
  Chengdu\\
  \texttt{wangyulonga@gmail.com} \\
  \texttt{\{chenxsh, qixuwang\}@scu.edu.cn} \\
}

\begin{document}
\maketitle

\begin{abstract}
Cloud-edge collaborative inference approach splits deep neural networks (DNNs) into two parts that run collaboratively on resource-constrained edge devices and cloud servers, aiming at minimizing inference latency and protecting data privacy. However, even if the raw input data from edge devices is not directly exposed to the cloud, state-of-the-art attacks targeting collaborative inference are still able to reconstruct the raw private data from the intermediate outputs of the exposed local models, introducing serious privacy risks. In this paper, a secure privacy inference framework for cloud-edge collaboration is proposed, termed CIS, which supports adaptively partitioning the network according to the dynamically changing network bandwidth and fully releases the computational power of edge devices. To mitigate the influence introduced by private perturbation, CIS provides a way to achieve differential privacy protection by adding refined noise to the intermediate layer feature maps offloaded to the cloud. Meanwhile, with a given total privacy budget, the budget is reasonably allocated by the size of the feature graph rank generated by different convolution filters, which makes the inference in the cloud robust to the perturbed data, thus effectively trade-off the conflicting problem between privacy and availability. Finally, we construct a real cloud-edge collaborative inference computing scenario to verify the effectiveness of inference latency and model partitioning on resource-constrained edge devices. Furthermore, the state-of-the-art cloud-edge collaborative reconstruction attack is used to evaluate the practical availability of the end-to-end privacy protection mechanism provided by CIS.
\end{abstract}

\section{Introduction}
With the advent of the Internet of Everything and the fifth-generation communication era, the decentralized and fragmented data generated by network edge devices is growing exponentially, and the demand for data transmission bandwidth is increasing\cite{mao2017survey}. Meanwhile, new scenarios such as industrial Internet and autonomous driving have created new demands for real-time data processing and security and privacy that traditional centralized cloud computing architectures can no longer effectively address\cite{wang2019ecass,kuang2021cooperative}. To address these challenges, edge computing has emerged with the advantage of being closer to the data side, sinking part of the computing and storage tasks from the center to the edge and breaking the bottleneck of the traditional network schemes through cloud-edge collaboration \cite{siriwardhana2021survey}.

Nevertheless, deploying some computationally intensive tasks in resource-constrained edge devices still faces significant computational latency and energy consumption. Kang et al. designed a lightweight cloud-edge collaborative inference framework, Neurosurgeon\cite{kang2017neurosurgeon}, to fine-grainedly partition the DNN network for the variation of data size and computation in each layer of the neural network. As shown in Figure \ref{fig1} for the AlexNet network as an example, the latency and output data size of each layer exhibit a large heterogeneity, which means that layers with higher latency may not necessarily output a larger amount of data. Based on this observation, Neurosurgeon reduces the total end-to-end execution latency as well as the energy consumption of edge devices by dividing the DNN into two parts, and offloading the computationally intensive part to the server at a lower transmission cost (only the intermediate results of the layer where the partition point is located need to be transmitted). Subsequently, more and more research on collaborative inference based on model partitioning \cite{li2019edge,zhang2021autodidactic,hu2019dynamic} have been proposed to further improve the performance and efficiency of such approaches, while ignoring the potential security and privacy issues.

The work of He et al. demonstrates the feasibility of data privacy attacks against cloud-edge collaborative inference systems \cite{he2020attacking}, even if the cloud only receives intermediate results offloaded by edge devices instead of the original data. The untrustworthy cloud can still easily and accurately recover sensitive original data from intermediate results by means of white-box and black-box attacks. In addition, the overfitting of the model also provides a side channel of data privacy leakage for cloud-edge collaborative inference \cite{yeom2018privacy}. The membership inference attack proposed by Shokri et al. \cite{shokri2017membership} can be exploited by malicious edge nodes to infer whether a certain data exists in the training set by querying the results of the black-box inference service.

To address these challenges, cryptography-based protocols \cite{liu2017oblivious,gilad2016cryptonets} have also been proposed to protect data privacy in the inference phase. However, complex and frequent cryptographic computations introduce significant computational and communication overheads while preserving privacy, which are infeasible to deploy on devices with constrained computational resources such as IoT devices. In addition to this, some scholars have leveraged differential privacy as a lightweight privacy-preserving strategy to achieve privacy preservation in machine learning by adding quantifiable noise to the model or output results through provable mechanisms. Among them, Mireshghallah et al. proposed the Cloak framework \cite{mireshghallah2020principled}, which maximizes privacy by employing an optimized Laplace distribution for obfuscation before sending the privacy data to the cloud, which maximizes privacy by minimizing the mutual information between the original input and the data sent to the cloud. However, direct perturbation of the original image can significantly damage the availability of the image, which in turn degrades the inference accuracy severely. Wang et al.\cite{wang2018not} proposed to add fine-calibrated noise to the intermediate output to achieve a differential privacy-preserving framework, which designs a noise training method to mitigate the impact of noise perturbation on inference accuracy. Similar to this work, an end-to-end collaborative inference privacy-preserving framework, Shredder\cite{mireshghallah2020shredder}, proposed by Mireshghallah et al. significantly reduces the amount of information in the communication data by learning to add noise to the distribution without changing the structure and weights of the pre-trained network, while being able to maintain inference accuracy.

In summary, the above studies do not reasonably combine cloud-edge collaborative inference with privacy-preserving mechanisms to effectively trade-off the conflicting problems between computational latency, privacy and availability. Therefore, in this paper, we propose a novel cloud-edge collaborative security inference framework, CIS (\underline{C}ollaborative \underline{I}nference \underline{S}hield), which aims to maximize the privacy strength with minimal impact on DNN accuracy, while being able to effectively tradeoff between usability and privacy of cloud-edge collaborative inference. The main contributions are summarized as follows.

\begin{itemize}
    \item A secure privacy inference framework CIS for cloud-edge collaboration is proposed. CIS supports adaptively partitioning the network for collaborative inference based on dynamically changing network bandwidth, aiming to fully release the computational power of edge devices. Meanwhile, the selection of a partition point fully considers the amount of information in offloading intermediate data and can effectively trade-off the total inference latency with the privacy of sensitive data at the edge.
    
    \item CIS provides a way to achieve differential privacy protection by adding refined noise to the intermediate layer feature maps offloaded to the cloud. Meanwhile, with a given total privacy budget, the budget is reasonably allocated by the size of the feature graph rank generated by different convolution filters, which makes the inference in the cloud robust to the perturbed data, thus effectively trade-off the conflicting problem between privacy and availability.
    
    \item We construct realistic cloud-edge collaborative inference computing scenarios to evaluate the effectiveness of inference latency and model partitioning on resource-constrained edge devices. Also, state-of-the-art cloud-edge collaborative reconstruction attacks from internal and external adversaries are used to evaluate the practical usability of the end-to-end privacy protection mechanisms provided by CIS.
\end{itemize}

 \begin{figure}
 \centering
  \begin{subfigure}
  \centering
  \includegraphics[width=0.7\linewidth]{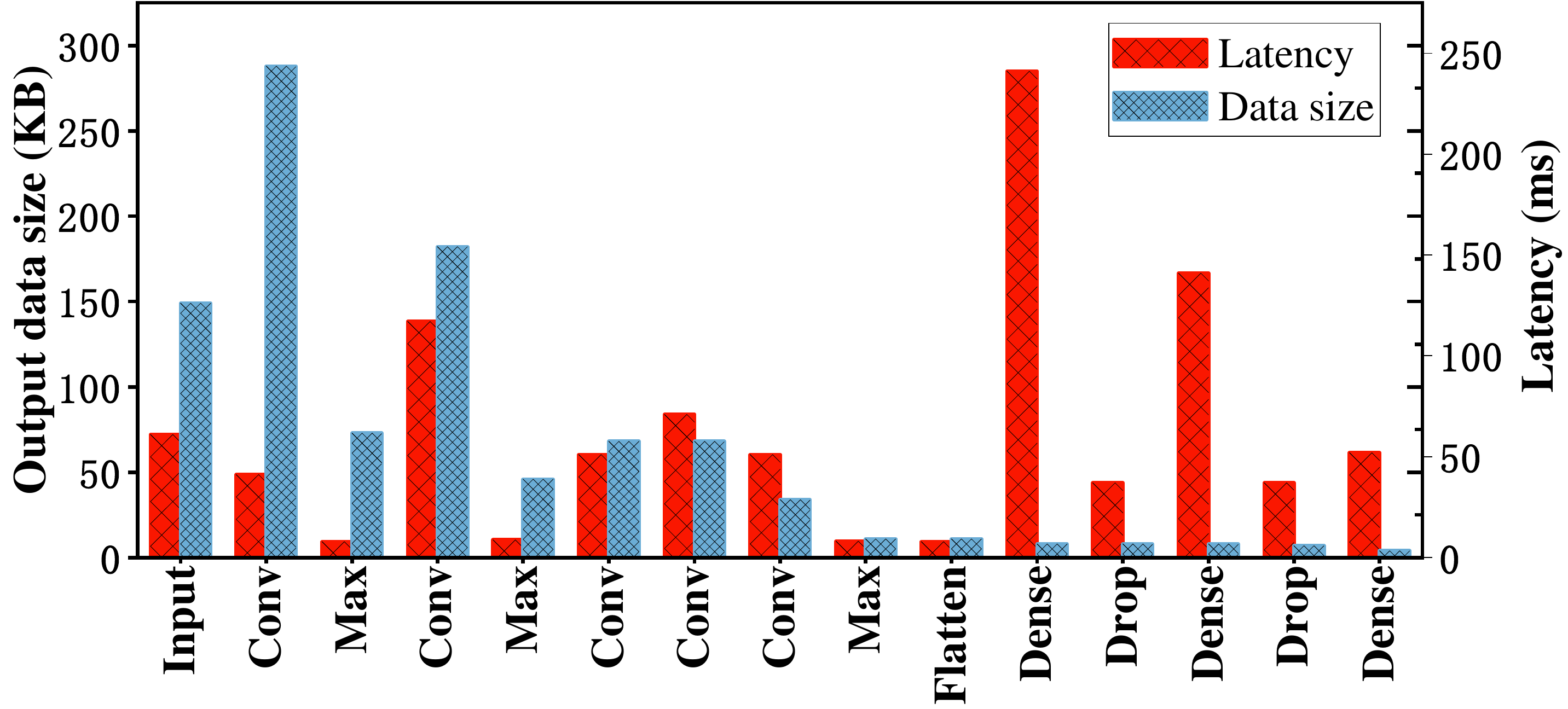}
  \caption{Output data size and execution delay for each layer in AlexNet network.}
  \label{fig1}
  \end{subfigure}
\end{figure}

\section{Preliminaries and Related Works}
\subsection{Differential privacy}
Differential Privacy (DP) \cite{dwork2006calibrating}, proposed by Dwork in 2006 as an alternative reliable privacy model, has been considered as a promising privacy-preserving strategy in machine learning in recent years.The definition of DP is based on a rigorous theoretical foundation that Privacy preservation in machine learning is achieved by adding quantifiable noise to the model or output results through provable mechanisms, and an elegant trade-off between privacy strength and usability can be made by adjusting the privacy budget \cite{zhao2022survey}.

Differential privacy guarantees that queries and accesses of any random algorithm on two adjacent datasets have similar output distributions, and an attacker cannot infer private information about an individual in the results of any query. A formal definition of differential privacy is shown below:
\begin{definition} ($\epsilon$-Differential Privacy {\rm \cite{dwork2006calibrating}}, $\epsilon$-DP)
For any two adjacent data sets $D$ and $D^{'}$, given a randomized mechanism algorithm $\mathcal{M}$ running on this adjacent data set, $S$ is a subset of all possible outputs generated by the $\mathcal{M}$ mechanism, when the following inequality holds
\begin{equation}
\label{eq-2-1}
Pr\left [ \mathcal{M}\left ( D\right ) \in S \right ]\leqslant e^{\epsilon}Pr\left [ \mathcal{M}\left ({ D}'\right ) \in S \right ]
\end{equation}
\label{def-2-1}
\noindent Then the random mechanism $\mathcal{M}$ is claimed to satisfy $\epsilon$-differential privacy, where the parameter $\epsilon$ is the privacy-preserving budget, and the smaller the $\epsilon$, the higher the privacy-preserving strength.

\end{definition}
Before introducing the specific privacy protection mechanism, we give the definition of global sensitivity:
\begin{definition} (global sensitivity {\rm \cite{dwork2014algorithmic}})
With a query function $f:\mathbb{N}^{\left |\mathcal{X} \right |}\rightarrow \mathbb{R}^{k}$ for any two adjacent data sets $D$ and $D^{'}$, the global sensitivity is defined as
\begin{equation}
\label{eq-2-2}
\triangle s\left ( f,\left \| \cdot  \right \|\right )= \max\limits_{d\left ( {D,{D}'}\right )=1}\left \|f\left ( D\right )-f\left ( {D}'\right ) \right \|
\end{equation}
\label{def-2-2}
\noindent where $\left \| \cdot \right \|$ is the distance metric, usually $l_{1}$ and $l_{2}$ norm.

\end{definition}

Based on the above mentioned definitions of differential privacy and global sensitivity, The Laplace Mechanism implements the $\epsilon$-differential privacy-preserving mechanism by adding random noise obeying the Laplace distribution ($Lap\left ( x\mid b\right )=\frac{1}{2b}exp\left ( -\frac{\left | x\right |}{b}\right )$) to the query result, as defined below.

\begin{definition}(Laplace mechanism {\rm\cite{dwork2014algorithmic}}
Given any query function $f:\mathbb{N}^{\left |\mathcal{X} \right |}\rightarrow \mathbb{R}^{k}$ with global $l_{1}$ sensitivity $\triangle s$, the Laplace mechanism is defined as

\begin{equation}
\label{eq-2-3}
\mathcal{M}_{L}\left ( \left ( x,f\left ( \cdot\right ),\epsilon \right )\right )=f\left ( x\right )+\left ( Y_{1},Y_{2},\cdots,Y_{k}\right )
\end{equation}
\label{def-2-3}
\noindent where $Y_{i}$ is drawn from $Lap\left(\frac{\triangle s}{\epsilon}\right)$ of independent identically distributed random variables, and the Laplace mechanism satisfies $\epsilon$-differential privacy.
\end{definition}

In addition to this, differential privacy has a very important property, namely the post-processing property \cite{dwork2014algorithmic}. After any processing of the output of a randomized algorithm satisfying differential privacy, the post-processing property still guarantees privacy protection of the same privacy strength. This property enables the application and generalization of differential privacy to complex machine learning algorithms with the following properties:
\begin{property}
If the random mechanism $\mathcal{M}:\mathbb{N}^{\left |\mathcal{X} \right |}\rightarrow R$ satisfies $\left(\epsilon,\delta\right)$-DP. Let $f:\mathcal{R}\rightarrow \mathcal{{R}'}$ be an arbitrary random mapping, then $f\circ \mathcal{M}:\mathbb{N}^{\left |\mathcal{X} \right |}\rightarrow {R}'$ still satisfies $\left(\epsilon,\delta\right) $-differential privacy.
\label{pro-2-1}
\end{property}

In fact, a complex computational task often does not correspond to a complex privacy-preserving mechanism, but rather the privacy budget is rationally allocated to the various steps of the complex task. The combination theorem of differential privacy gives a solution in this case for computing the privacy-preserving strength and performance of the entire complex computational task.

\begin{theorem} (Parallel combination theorem {\rm \cite{mcsherry2009privacy}})
Suppose that given a data set $D$ with $k$ mutually disjoint subset divisions $\left\{D_{1},D_{2},\cdots,D_{k}\right\}$, the privacy mechanism $\mathcal{M}_{i}:\mathbb{N}^{\left |\mathcal{X} \right |}\rightarrow \mathcal{R}_{i}$ is performed separately for each subset, satisfying $\epsilon_{i}$-differential privacy for any $i\in \left [ k \right ]$. Then the mechanism $\mathcal{M}_{\left [ k\right ]}=\left \{ \mathcal{M}_{1}\left(D_{1}\right),\mathcal{M}_{2}\left(D_{2}\right),\cdots,\mathcal{M}_{k }\left(D_{k}\right)\right \}:\mathbb{N}^{\left |\mathcal{X} \right |}\rightarrow \sum_{i=1}^{k}\mathcal{R}_{i}$ is satisfying $max \epsilon_{i}$-difference Privacy.
\label{theo-2-1}
\end{theorem}

\subsection{Cloud-edge collaborative inference and privacy enhancement}
Due to the constraints of limited computing resources and energy consumption of edge devices, edge devices need to make decisions to offload a portion of computing tasks to the cloud to collaboratively complete tasks \cite{mach2017mobile}. The cloud-edge collaborative inference research proposed in recent years
research \cite{kang2017neurosurgeon,teerapittayanon2017distributed,ko2018edge,zhang2021deepslicing,zhang2021autodidactic} overcome the significant communication overhead and potential privacy leakage by cutting the model efficiently, where the former part of the computation remains on the edge device and the resulting intermediate computation results are offloaded to the cloud to complete the remaining computational tasks. Further, Hu et al \cite{hu2019dynamic} propose the model partitioning method for directed acyclic graph (DAG) topological DNN models, but it is inefficient for extensive DNN model partitioning. To address the shortcomings of offline partitioning methods, \cite{zhang2021deepslicing,banitalebi2021auto} proposed an adaptive online partitioning method to achieve better results by adjusting the partitioning strategy in real time. 
\cite{manasi2020neupart,xu2020energy} studied the energy consumption issue under collaborative inference based on model partitioning method to achieve the lowest energy consumption while guaranteeing the delay requirement.

However, the work of Zecheng He et al. \cite{he2020attacking,he2019model} shows that an untrusted cloud can still easily and accurately recover sensitive data from intermediate values in a limited attack background, without even accessing the edge model. Accordingly, many scholars have started to conduct a lot of research on security privacy protection for collaborative inference at the cloud edge. Accordingly, many scholars have started to conduct a lot of research on security privacy protection for collaborative inference at the cloud edge. The most representative studies are adding quantifiable noise satisfying differential privacy to the intermediate output values to achieve data privacy protection for edge devices \cite{mireshghallah2020shredder,wang2020differential,mao2018learning}. However, these perturbations eventually lead to degradation of model inference accuracy, and the balance between good privacy and usability still faces a great challenge.

 \begin{figure}
 \centering
  \begin{subfigure}
  \centering
  \includegraphics[width=0.7\linewidth]{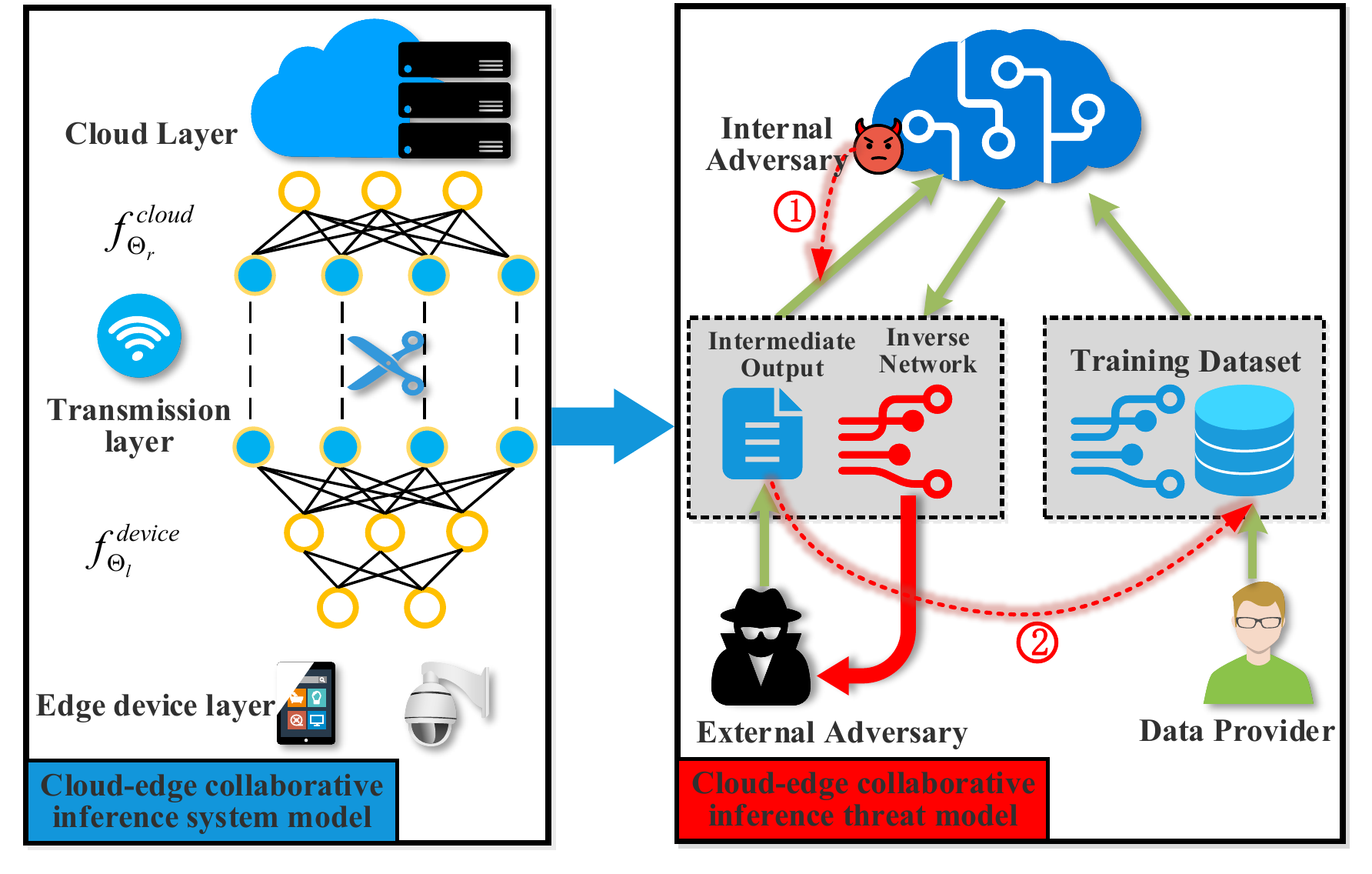}
  \caption{System model and threat model of CIS scheme.}
  \label{fig2}
  \end{subfigure}
\end{figure}

\section{Problem Statement}
This section analyzes and elaborates on the system model and the threat model faced by the CIS framework, and specifies the design goals of the solution in this section.
\subsection{System Model}
In this work, we consider how to tackle the conflicting problems between collaborative inference performance, privacy and availability in the scenario of cloud-side collaborative deep learning inference oriented. As shown in the left side of Figure. \ref{fig2}, the system model of CIS proposed in this work consists of three main components: Edge device layer $\mathbb{E}$, Cloud layer $\mathbb{C}$, and Transmission layer $\mathbb{T}$. Initially, in the collaborative inference architecture, the original $n$-layer DNN network $f_{\Theta} = \left \{ f_{\theta_1}\circ  f_{\theta_2}\cdots\circ  f_{\theta_n}\right \}$ is cut into two parts at layer $m$: $f^{device}_{\Theta_l}=\left \{f_{\theta_1}\circ f_{\theta_1}\cdots\circ f_{\theta_m} \right \}$ and $f^{cloud}_{\Theta_r}=\left \{f_{\theta_m+1}\circ f_{\theta_m+2}\cdots\circ f_{\theta_n} \right \}$, where $\mathbb{E}$ performs the first half of the DNN network and sends the intermediate computation result $v_m=f^{device}_{\Theta_l}\left(x\right)$ from layer $m$ to the remote $\mathbb{C}$ via $\mathbb{T}$. Subsequently, the $\mathbb{C}$ performs the remaining second half of the network computation and returns the final inference result $y=f^{cloud}_{\Theta_r}\left(v_m\right)$.

The selection of the DNN network partition point needs to consider many different factors to determine the optimal strategy. The most important performance metric is the overall inference delay, which mainly consists of the inference computation time of the edge devices, the inference computation time on the cloud, and the network transmission time. The CIS system model proposed in this work will consider the variation of network communication quality in real scenarios, the amount of computation in different layers, and the size of data output from different layers to select the optimal model cutting position to maximize the overall system performance. In addition, another issue that needs to be addressed in the CIS system model is the privacy of the data.

\subsection{Threat Model}
As illustrated in the right half of Figure. \ref{fig2}, the threat model for the collaborative cloud-edge deep learning inference service considered in this work involves three computational entities, namely the Edge device layer $\mathbb{E}$ and the Cloud service layer $\mathbb{C}$ (internal adversary), and the Malicious edge node $\mathbb{M}$ (external adversary), between which computations are connected by the intermediate output $v_m$ of the partition layer $m$. 

We assume that the $\mathbb{E}$ is trustworthy and its will not disclose any information to other computing entities voluntarily. However, the $\mathbb{C}$ is an honest-but-curious entity, i.e., although $\mathbb{C}$ is curious about the private data of the edge devices, it will strictly follow the predefined computation protocols and will not interfere with the whole collaborative inference computation process. Similarly, a malicious edge node $\mathbb{M}$ adheres strictly to the predefined computation protocols and is more restrictive, and does not have any prior knowledge except through legitimate query requests. However, these advanced adversaries from external and internal entities can still perform other computations to infer or reconstruct the input data privacy of the edge device \cite{he2020attacking}. Specifically, the model of internal and external threats to CIS and the associated assumptions are as follows:
\begin{itemize}
\item \textbf{White-box Reconstruction Attack, WRA: } It is assumed that $\mathbb{C}$ does not have any a priori information about the original privacy data $x$ except for the acquisition of the intermediate output $v$ of the partition layer uploaded by $\mathbb{E}$. In addition, the cloud-edge collaborative inference service needs to share the same deep inference network, so it is also assumed that $\mathbb{C}$ has information about the network structure and parameters $f_{\Theta}$ of the model. The white-box reconstruction attack can be formally defined as $\hat{x}=WRA\left(f^{device}_{\Theta_l}\left(x\right),f_{\Theta}\right)$.

\item \textbf{Black-box Inverse-Network attack, BINA: } Assume that $\mathbb{M}$ does not have any a priori information about the original private data $x$, as well as no a priori information about the network and parameters $f_{\Theta}$ of the model, other than the ability to use the cloud-edge collaborative inference service through legitimate query requests normally. However, the intermediate data uploaded to the edge devices are available via bypass. the BINA attack can be formally defined as $\hat{x}=f^{-1}_{\Theta_l}\left(f^{device}_{\Theta_l}\left(x\right)\right)$ by training an inverse model $f^{-1 }_{\Theta_l}$ to reconstruct the input data from the intermediate results.

\end{itemize}

\subsection{Design Objectives}
According to the system model and threat model proposed above, the goal of CIS is to design and implement a cloud-edge collaborative inference framework with high privacy, high precision and low latency, which can resist advanced threats from internal and external. Combined with security requirements, the specific design goals of CIS are as follows:

\textbf{(1) Inference delay:} CIS supports adaptively cutting the network to achieve collaborative inference according to the dynamically changing network bandwidth in order to fully utilize the computing power of edge devices and minimize the total delay of cloud-side collaborative inference.

\textbf{(2) Privacy}: CIS supports adding refined noise based on differential privacy to effectively resist advanced black-box and white-box reconstruction attacks, and the impact of privacy protection on model accuracy is within a reasonable range.

\textbf{(3) Usability}: The protection mechanism and model partitioning mechanism of CIS can be easily applied to existing deep networks without any structural and parametric modification to the network model.
\section{ Proposed Method}

\subsection{Cloud-Edge collaborative inference acceleration based on model partitioning} %
To formally define the model partition problem, the original DNN model can be transformed into a directed linked list $\mathcal{L}=\left \{\mathcal{V},\mathcal{E} \right \}$, where $\mathcal{V}=\left \{\mathscr{l}_1,\mathscr{l}_2,\cdots,\mathscr{l}_n\right \}$ denotes the definition of each layer of the DNN model as the set of vertices in $\mathcal{L}$, $\mathcal{E}$ denotes the set of dependencies between layers, and $\left \langle \mathscr{l}_i,\mathscr{l}_j \right \rangle \in \mathcal{E}$ denotes that the computed output data of $\mathscr{l}_i$ layer will be transferred to $\mathscr{l}_j$ layer as input. On the basis of the formal definition of DNN as a directed linked list $\mathcal{L}$, we will give some other important definitions.

\begin{definition} (Model partitioning problem)

\noindent Given a $n$-layer DNN model with a directed chain list $\mathcal{L}=\left \{\mathcal{V},\mathcal{E} \right \}$, the model partitioning problem can be defined as taking some vertex $\mathscr{l}_m$ in $\mathcal{V}$ as the partition point, and the edge $\left \langle \mathscr{l}_m,\mathscr{l}_{m+1}\right \rangle$ is cut and $\mathcal{L}$ is split into two parts, where $\mathcal{L}_{edge}=\left \{\mathscr{l}_1,\mathscr{l}_ 2,\cdots,\mathscr{l}_m \right \}$ denotes the layer that will perform the computation at the edge device, $\mathcal{L}_{cloud}=\left \{\mathscr{l}_{m+1},\mathscr{l}_{m+2},\cdots,\mathscr{l}_n \right \}$ denotes the layer that will be offloaded to the cloud to perform the computation.

\label{def-1}
\end{definition}

The latency of cloud-edge collaborative inference mainly includes computation latency and transmission latency, both of which can be obtained by monitoring computation and network resources respectively in the offline configuration stage, and subsequently computed by analyzing the fixed model layer by layer, as shown in the upper left part of Figure. \ref{fig3}. The transmission latency includes the latency of uploading data $T^t_{up}$ and the latency of returning results $T^t_{down}$. The latency of transmitting data from the edge devices to the cloud server through the wireless network depends on the data rate of the network $R\left ( t\right )$, which can be calculated by the Shannon-Hartley theorem \cite{goldsmith2005wireless} as follows.

 \begin{figure}
 \centering
  \begin{subfigure}
  \centering
  \includegraphics[width=0.7\linewidth]{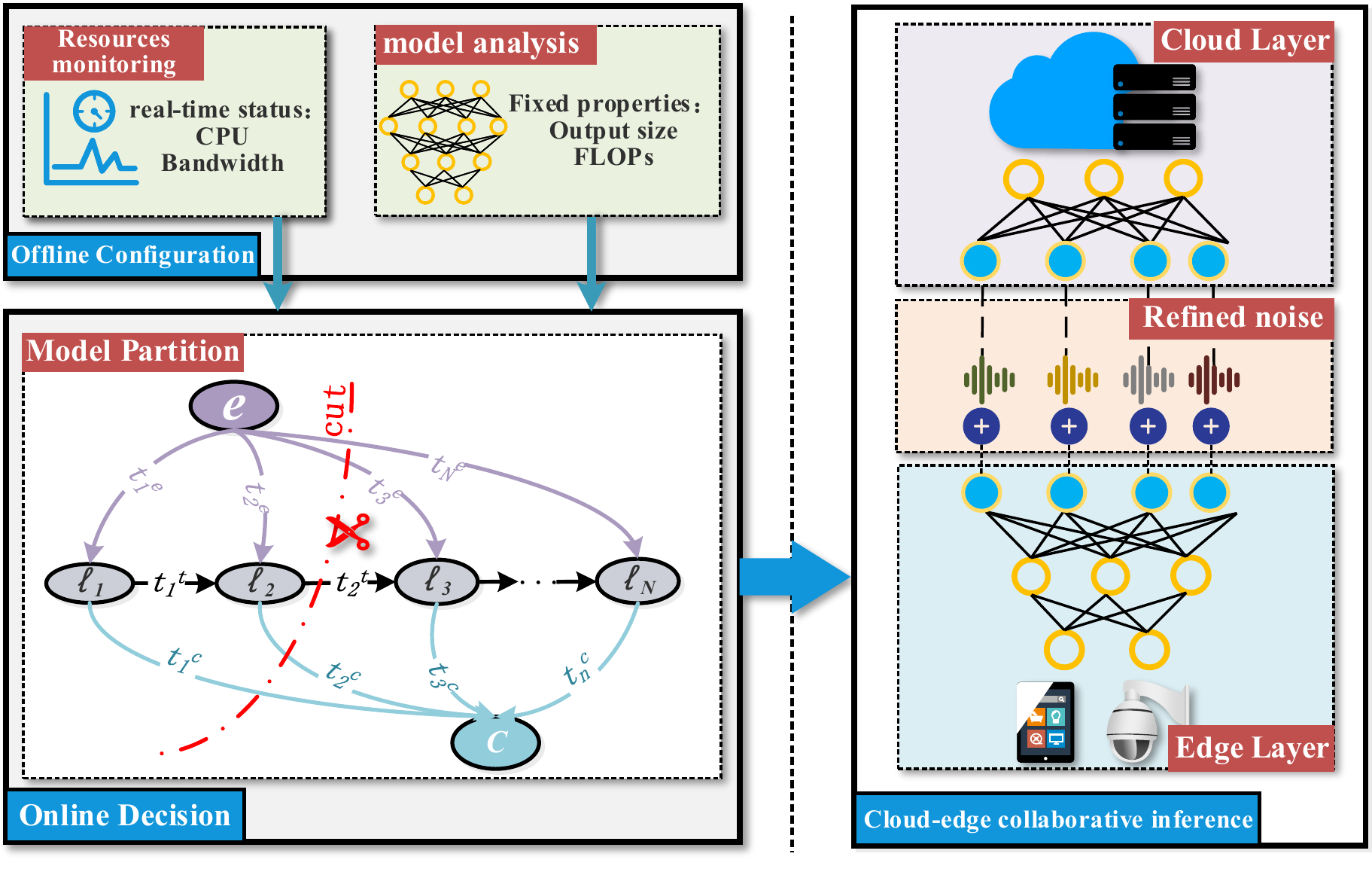}
  \caption{CIS System Framework.}
  \label{fig3}
  \end{subfigure}
\end{figure}

\begin{equation}
\label{eq-1}
R\left ( t\right )=B_w\log_2\left ( 1+\frac{P\left ( t\right )g\left ( t\right )}{\sigma^{2}+I\left ( t\right )}\right )
\end{equation}
where $B_w$ and $g\left ( t\right )$ denote the bandwidth and flat fading channel gain at the instantaneous $t$ moment, respectively; $P\left ( t\right )$ denotes the transmission power of the edge device, $\sigma^{2}$ denotes the noise power of the edge device, and $I\left( t\right)$ denotes the inter-area interference power.

Thus, the specific definition of transmission latency for cloud-edge collaborative inference can be expressed as follows: 
\begin{definition} (Transmission Latency)
Given a $n$-layer DNN model with a directed linked list $\mathcal{L}=\left \{\mathcal{V},\mathcal{E} \right \}$, a partition point of $\mathscr{l}_m$ with the output size of $D_m$, and $D_r$ being the data size of the inference result. The transmission inference of the cloud-edge collaboration is defined as follows:
\begin{equation}
\label{eq-2}
T^t = T^t_{up}+T^t_{down} = \frac{D_m}{R\left ( t\right )}+\frac{D_r}{R\left ( t\right )}
\end{equation}
\label{def-5-2}
\end{definition}

Another important percentage of inference latency is the computation latency of each layer of the DNN. Since the network structure and the number of parameters of the model during DNN inference are frozen, and the input size of the edge devices is also fixed in size. Therefore, we refer to the method mentioned in the literature \cite{9155237} to estimate the computational latency of each layer by counting the data volume of floating-point operations per second (FLOPs) of different types of layers in DNN. The amount of FLOPs computed for the convolutional layer $F_{conv}$ and the fully connected layer $F_{fully}$ can be expressed as:

\begin{align}
\label{eq-3}
F_{conv}&= 2HW\left(C_{in}K^2+1\right)C_{out} \\
F_{fully}&=\left(I\cdot O+O\cdot \left(I-1\right)\right) =\left(2I-1\right)O
\end{align}
\noindent where $H$, $W$ and $C_{in}$ denote the height, width and number of channels of the input feature map, respectively. $K$ denotes the size of the convolution kernel, and $C_{out}$ denotes the number of channels of the output convolution layer. $I$ and $O$ denote the input and output dimensions of the fully connected layer, respectively. Note that the activation layer is assumed to be a Rectified linear unit ReLU, which has negligible execution time compared to the dot product computation of the convolution and fully connected layers.

\begin{definition} (Computation Latency )
Given the $i_{th}$ layer of a DNN model, according to the type of the layer (convolutional layer or fully-connected layer), its execution delay on edge devices and cloud servers can be expressed as:
\begin{align}
\label{eq-5}
t^e_i = \left(F_{conv}|F_{fully}\right)/P_{edge} \\
t^c_i = \left(F_{conv}|F_{fully}\right)/P_{cloud}
\label{def-3}
\end{align}
\end{definition}
\noindent where $P_{edge}$ and $P_{cloud}$ denote the floating-point computing power of the edge device or cloud server, respectively, which can be obtained from the CPU or GPU specifications.

Based on the above definitions and analysis, CIS can obtain statistics on edge computation latency, transmission latency and cloud computation latency for each layer in the offline configuration phase: $T^c_{edge}=\left\{t^e_1,t^e_2,\cdots,t^e_n\right\}$, $T^t=\left\{t^t_1,t^t_2,\cdots,t^t_n\right\}$, and $T^c_{cloud}=\left\{t^c_1,t^c_2,\cdots,t^c_n\right\}$. As shown in the lower left part of \ref{fig3}, two virtual vertices $e$ and $c$ are constructed to represent the edge layer and the cloud service layer, respectively. The above offline statistical latency information is combined with the DNN model directed chain table $\mathcal{L}$ (defined in Eq.\ref{def-1}) and converted into a directed acyclic graph $\mathcal{G}=\left \{\mathcal{V},\mathcal{E}' \right \}$, where $\mathcal{E}'=\mathcal{E}\cup \left \{ \left \langle e,\mathscr{l}_i\right \rangle,\left \langle \mathscr{l}_i,c\right \rangle \right \}_{i=1}^{n} $. 
Towards the collaborative inference scenario, the optimal segmentation point selection for DNN networks needs to satisfy the minimization of the total inference delay, defined as follows:

\begin{definition} (Total Inference Latency )
Given a DNN model DAG graph $\mathcal{G}=\left \{\mathcal{V},\mathcal{E}'\right \}$ with the partition point $\mathscr{l}_m$ for $n$ layers, then the total inference latency for cloud-edge collaboration inference is defined as follows: 
\begin{equation}
\label{eq-7}
T_{total} = T^t\left(\mathscr{l}_m\right)+\sum_{\mathscr{l}_i\in \mathcal{L}_{edge}}t_i^e +\sum_{\mathscr{l}_j\in \mathcal{L}_{cloud}}t_j^c 
\end{equation}
\label{def-4}
\end{definition}
\noindent where $t_i^e$ and $t_j^c$ denote the computational latency of the corresponding layers on the edge devices and cloud servers, respectively.

Comprehensive analysis of the above, the model partition for collaborative inference in CIS is shown in algorithm \ref{algorithm1}. First, by monitoring the network transmission rate and analyzing the fixed properties of the model in the offline configuration phase, the computational delay of each layer in the edge devices and cloud servers, respectively, and the transmission delay of different layers can be predicted in advance (lines 1-3 of the algorithm \ref{algorithm1}). In line 4 of the algorithm \ref{algorithm1}, the original DNN network can be converted into a weighted directed acyclic graph $\mathcal{G}=\left \{\mathcal{V},\mathcal{E}'\right \}$ based on these statistical information. Finally, each layer is considered as a partition layer to count the total inference delay separately, and the layer with the minimum total inference delay is selected as the optimal splitting layer (lines 6-12 of the algorithm \ref{algorithm1}), which will be deployed in the cloud-edge collaborative inference environment respectively. It is worth noting that the dynamically changing network bandwidth affects the selection of optimal segmentation points. Therefore, CIS will constantly monitor the changes of network resources and can implement adaptive partitioning of DNN models to maximize the inference performance of cloud-edge collaboration.

\begin{algorithm}[h]
\caption{Model partition algorithm for cloud-edge collaborative inference.}
\label{algorithm1}
\hspace*{0.02in} {\bf Input:}
$f_{\Theta} = \left \{ f_{\theta_1}\circ  f_{\theta_2}\cdots\circ  f_{\theta_n}\right \}$, $n$-layer DNN network; $D=\left\{D_1,D_2,\cdots,D_n\right\}$, Output data size per layer; $R\left(t\right)$, the current network transmission rate; $P_{cloud}$ and $P_{edge}$, Floating-point computing power for cloud and edge devices.\\
\hspace*{0.02in} {\bf Output:}
$\mathscr{l}_{best}$, optimal model partition point; $T_{best}$, total inference time latency.
\begin{algorithmic}[1]
\State $T^c_{edge}=\left\{t^e_1,t^e_2,\cdots,t^e_n\right\} \leftarrow f^{c}_{edge}\left(f_{\Theta},P_{edge}\right)$
\State $T^c_{cloud}=\left\{t^c_1,t^c_2,\cdots,t^c_n\right\} \leftarrow f^{c}_{cloud}\left(f_{\Theta},P_{cloud}\right)$
\State $T^t=\left\{t^t_1,t^t_2,\cdots,t^t_n\right\} \leftarrow f^{t} \left(D,R\left(t\right)\right)$
\State 	$\mathcal{G}=\left \{\mathcal{V},\mathcal{E}'\right \} \leftarrow DAG\left(f_{\Theta},T^c_{edge},T^c_{cloud},T^t \right)$	\State 	$T_{best} = +\infty$

\For{$i = 1$ to $n$} 
    \State 	$\mathcal{L}_{edge}=\left\{\mathscr{l}_1,\mathscr{l}_2\cdots \mathscr{l}_i\right\},\mathcal{L}_{cloud}=\left\{\mathscr{l}_{i+1},\mathscr{l}_{i+2}\cdots \mathscr{l}_n\right\} \leftarrow Cut\left(\mathcal{G},\mathscr{l}_{best}\right)$ 
    \State 	$T_{total} = T^t\left(\mathscr{l}_i\right)+\sum_{\mathscr{l}_i\in \mathcal{L}_{edge}}t_i^e +\sum_{\mathscr{l}_j\in \mathcal{L}_{cloud}}t_j^c $
    \If{$T_{total} < T_{best}$}
        \State $T_{best}=T_{total}$ 	
        \State $\mathscr{l}_{best}\leftarrow \mathscr{l}_{i}$
    \EndIf
\EndFor     
\State \Return $T_{best},\mathscr{l}_{best}$

\end{algorithmic}

\end{algorithm}

\subsection{A privacy-enhancing mechanism for cloud-edge collaborative inference}
As described earlier in the threat model, CIS employs a cloud-edge collaborative inference schema to keep sensitive data from edge devices out of the local area, mitigating privacy security issues to some extent. However, it still faces white-box and white-black reconstruction attacks by advanced internal and external adversaries through the intermediate layer output results and legitimate queries. Therefore, as shown in the right part of Figure. \ref{fig3}, we propose a privacy-enhancing mechanism for cloud-edge collaboration, Collaborative-DP, which injects refined Laplace noise satisfying $\epsilon$-DP when the edge devices upload the intermediate output of the partition layer, thus satisfying the privacy-preserving enhancement of cloud-edge collaborative inference and minimize the degradation of inference accuracy.

First of all, as shown in the Laplace mechanism satisfying $\epsilon$-DP given by the definition \ref{def-2-3}, the added noise needs to be sampled from the distribution $Y\sim Lap\left(\frac{\triangle s}{\epsilon}\right)$, where $\epsilon$ is the privacy budget and $\triangle s$ is the global sensitivity.  However, in the cloud-edge collaborative inference scenario, the global sensitivity of the partition layer $\mathscr{l}_m$ is difficult to estimate without any a priori bound. Overly conservative estimation of $\triangle s$ will add too much noise to the output representation and will reduce the accuracy of subsequent inference in the cloud. Similar to the related work \cite{abadi2016deep,wang2018not}, Collaborative-DP employs an 
norm clipping of the uploaded intermediate results to a fixed bound as a way to estimate the global sensitivity. Specifically, for sensitive input $x$ from any edge device, an infinite norm is applied to clip the intermediate output result $v_m=f^{device}_{\Theta_l}\left(x\right)$ of the partition layer $\mathscr{l}_m$ separately for each channel:

\begin{equation}
\label{eq-8}
v'_m \left [ i\right ] =\frac{v_m \left [i \right ]}{\max\left(1,\frac{\left \|v_m \left [ i\right ]\right \|_{\infty}}{\boldsymbol{C_m}}\right)}, for\ i = 1,2,\cdots,k, \ \ \ v_m \in \mathbb{R}^{k*w*h}
\end{equation}
\noindent where $\boldsymbol{C_m}$ is the clipping threshold, which in practice can be set to the median of the infinite norm of the output tensor of the partition layer $\mathscr{l}_m$ during training. $k,w,h$ denote the dimensions of the number of channels, width and height of the intermediate result (tensor), respectively.

After clipping $v_m$ by infinite norm, when $\left \|v_m\left [ i\right ]\right \|_{\infty}\leq \boldsymbol{C_m}$, the value of $v_m\left [ i\right ]$ will be preserved, and vice versa the value of $v_m\left [ i\right ] $ value will be clipped to $\boldsymbol{C_m}$. Hence, the global sensitivity $\triangle s$ can be estimated as $2\boldsymbol{C_m}$, and the clipping threshold $\boldsymbol{C_m}$ is independent of the input and does not reveal any sensitive information. In brief, we can add the independent identically distributed noise sampled from the distribution $Y\sim Lap\left(\frac{2\boldsymbol{C_m}}{\epsilon}\boldsymbol{I}\right),\boldsymbol{I}\in  \mathbb{R}^{w*h}$ directly to the $v'_m\left [ i\right ]$ after clipping to ensure that the cloud-edge collaborative inference algorithm satisfies the $\epsilon$-DP. Unlike adding noise to the final output of the deterministic function, injecting noise during the transformation is more flexible and more suitable for the stacking structure of DNN networks.

Given a total privacy budget $\epsilon$, it is crucial to further refine the allocation of privacy budget and generation noise, which has a performance impact on inference accuracy. Inspired by the work of Lin et al. \cite{lin2020hrank} that the average rank of feature maps generated by a single filter is always the same, regardless of the number of image batches received by the DNN. Consequently, a small batch of input images ($g \approx 500$) can be utilized to accurately estimate the expectation of the feature map rank. The high rank of the feature map (i.e., the defined intermediate layer output $v_m$) reflects the magnitude of the amount of information extracted by the different convolution filters in the current partition layer $\mathscr{l}_m$. Therefore, the proposed Collaborative-DP allocates the privacy budget by the ratio of the rank of the feature submap $v_m\left [ i\right ]$ from different channels to the rank of all feature maps. Specifically, feature submaps with higher rank contribute more to the inference accuracy of the model and can be assigned a higher privacy budget (corresponding to less noise), thus achieving a tradeoff between privacy and availability for collaborative inference. The schematic of privacy budget allocation based on feature graph rank is shown in Fig. \ref{fig4}, and the related calculations are as follows: 

\begin{align}
\label{eq-9}
\epsilon_i = \epsilon \cdot \frac{\textbf{Rank}\left (v_m\left [ i\right ]\right)}{\sum_{j=1}^{k}\textbf{Rank}\left (v_m\left [ j\right ]\right )},\ \ \ where\notag \\
\textbf{Rank}\left (v_m\left [ i\right ]\right ) \approx  
\frac{1}{g}\sum_{t=1}^{g}\textbf{SVD}\left(v_m\left [ i\right ],t\right),\ \ \sum_{i=1}^{k}\epsilon_i = \epsilon
\end{align}
\noindent where $\textbf{Rank}\left(\cdot \right)$ estimates the expectation of the feature map rank by taking $g$ inputs, $v_m\left [ i\right ]$ denotes the feature map generated by the $i$th filter ($k$ in total), and $\textbf{SVD}\left(\cdot \right)$ denotes the rank of the feature map obtained by singular value decomposition.

 \begin{figure}
 \centering
  \begin{subfigure}
  \centering
  \includegraphics[width=0.7\linewidth]{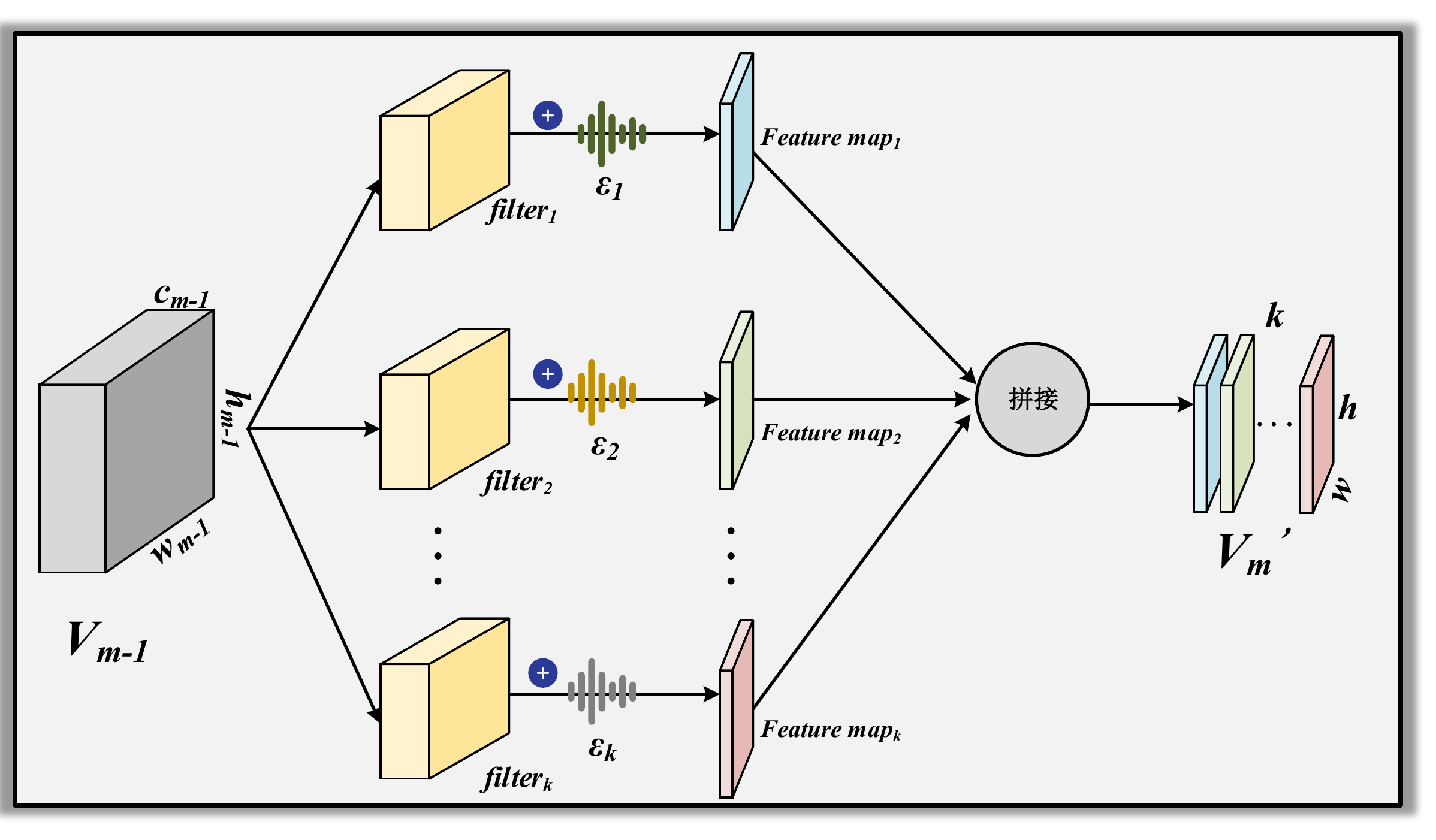}
  \caption{Schematic of privacy budget allocation based on feature map rank.}
  \label{fig4}
  \end{subfigure}
\end{figure}

Comprehensive analysis above, we give the detailed design of the adaptive privacy preservation cloud-edge collaborative inference algorithm, termed Collaborative-DP. The first line of the algorithm indicates that the privacy budget allocation based on the feature map rank is implemented according to the formula \ref{eq-9}; subsequently, the refined noise is added to the intermediate output by the edge devices at the partition layer in proportion to the privacy budget (lines 3-7 of the algorithm); finally, line 9 of the algorithm indicates that the remote cloud receives the intermediate results with the added noise and then completes the remaining inference network and outputs the results.

\begin{algorithm}[htb]
\caption{Adaptive privacy preserving cloud-edge collaborative inference algorithm (Collaborative-DP).}
\label{algorithm2}
\hspace*{0.02in} {\bf Input:}
$x$, the sensitive data input from the edge device; $D=\left\{x_i \right\}_{i=1}^g$, the training set; $f^{device}_{\Theta_l}\left(\cdot\right)$, the network layer executed by the edge device; $f^{cloud}_{\Theta_r}\left(\cdot \right )$, the network layer executed by the cloud server; $\mathscr{l}_m$, the network partition layer (noise addition layer); the number of convolution filters in $\mathscr{l}_m$; $\boldsymbol{C_m}$, the clipping threshold; $\left\{\epsilon_i \right \}_{i=1}^k$, the privacy budget assigned to different feature submaps.\\
\hspace*{0.02in} {\bf Output:}
$y$, inference result.
\begin{algorithmic}[1]
\State 	$\left\{\epsilon_i \right\}_{i=1}^k \leftarrow AllocateBudget\left(D,f^{device}_{\Theta_l}\left(\cdot\right)\right)$ 
\State 	$v_m \leftarrow f^{device}_{\Theta_l}\left(x\right)$ 
\For{$i = 1$ to $k$} 
    \State 	$v'_m\left [ i\right ]\leftarrow v_m \left [i \right ] / \max\left(1,\frac{\left \|v_m \left [ i\right ]\right \|_{\infty}}{\boldsymbol{C_m}}\right)$
    \State $\hat{v}'_m\left [ i\right ]\leftarrow v'_m\left [ i\right ] + Laplace\left(\frac{2\boldsymbol{C_m}}{\epsilon_i}\boldsymbol{I}\right)$ 
 
\EndFor 
\State 	$y \leftarrow f^{cloud}_{\Theta_r}\left(\hat{v}'_m\right)$	

\State \Return $y$
\end{algorithmic}

\end{algorithm}

Next, it is necessary to prove whether the algorithm \ref{algorithm2} satisfies the strict differential privacy guarantee.
\begin{theorem}
Adaptive privacy preserving cloud-edge collaborative inference algorithm (Collaborative-DP) satisfies $\sum_{i=1}^{k}\epsilon_{i}$-differential privacy.
\label{the-2-2}
\end{theorem}
\begin{proof}
 Initially, we consider the privacy-preserving case of a single convolutional filter $filter_i$. We use $f_{\theta_m}^{i}\left(\cdot\right)$ to denote the $i$th convolutional computation function of the partition layer $\mathscr{l}_m$.  Assuming that $D$ and $D'$ are adjacent data sets, the Laplace mechanism can be defined as the random function $\mathcal{M}_i\left(x,f_{\theta_m}^{i}\left(\cdot\right),\epsilon_i\right)=f_{\theta_m}^{i}\left(x\right)+Laplace\left(\frac{2\boldsymbol{C_m}}{\epsilon_i}\boldsymbol{I}\right)$, where the global sensitivity $\left |\Delta s\right | = 2\boldsymbol{C_ m}$. Therefore, for any output data point $t \in \mathbb{R}^{w*h}$ of the random function $\mathcal{M}_L$, we have
 \begin{align}
\label{eq-5-10}
\frac{Pr\left ( \mathcal{M}_i\left ( D\right )\right )=t}{Pr\left ( \mathcal{M}_i\left ( D'\right )\right )=t}&=\prod_{j=1}^{w*h}\frac{exp\left ( \frac{-\epsilon_i\left | t_j-f_{\theta_m}^{i}\left(D\right)_j\right |}{2\boldsymbol{C_m}}\right )}{exp\left ( \frac{-\epsilon_i\left | t_j-f_{\theta_m}^{i}\left(D'\right)_j\right |}{2\boldsymbol{C_m}}\right )}\notag
\\&=exp\left(\frac{\epsilon_i \left [-\sum_{j=1}^{w*h}\left | t_j-f_{\theta_m}^{i}\left(D\right)_j\right | - \sum_{j=1}^{w*h}\left | t_j-f_{\theta_m}^{i}\left(D'\right)_j\right |  \right ]
}{2\boldsymbol{C_m}}\right)\notag
\\&\leq exp\left ( \frac{\epsilon_i \left \| f_{\theta_m}^{i}\left(D\right)-f_{\theta_m}^{i}\left(D'\right)\right \|_1}{2\boldsymbol{C_m}}\right )\notag
\\&=exp\left(\epsilon_i\right)
\end{align}
Hence the convolution computation of a single filter $filter_i$ in the algorithm is satisfying $\epsilon_i$-differential privacy. As shown in Figure \ref{fig4}, all the convolutional computations of the partition layer $\mathscr{l}_m$ can be considered as a set of privacy mechanisms performed sequentially on the same data source $v_{m-1}$. Therefore, the combination mechanism (theorm \ref{theo-2-1}) of the partition layer $\mathcal{M}_{\left [ k\right ]}=\left \{ \mathcal{M}_{1},\mathcal{M}_{2},\cdots,\mathcal{M }_{k}\right \}$ is satisfied by $\sum_{i=1}^{k}\epsilon_{i}$-differential privacy.
Finally, according to the post-processing property of differential privacy (property \ref{pro-2-1}), for the remaining inference calculations $f^{cloud}_{\Theta_r}$ of the cloud server, $f^{cloud}_{\Theta_r}\circ \mathcal{M}_{\left [ k\right]}$ still satisfies $ \sum_{i=1}^{k}\epsilon_{i}$-differential privacy. Therefore, the adaptive privacy preserving cloud-edge collaborative inference algorithm (Collaborative-DP)satisfies $\sum_{i=1}^{k}\epsilon_{i}$-differential privacy.
\end{proof}

\section{ Evaluation}
In this section, the experimental setup of the CIS system is first presented, including the environment configuration and performance comparison baseline, the evaluation metrics, and the models and datasets used. Subsequently, CIS will be evaluated in several aspects such as inference latency and accuracy, privacy strength and availability, and success rate against different types of attacks, respectively.
\subsection{Experiment Setup}
\subsubsection{Experimental environment and configuration}

This work builds a real hardware experimental platform to evaluate the feasibility of the CIS system, where the edge device is a Jetson NANO mobile platform developed by NVIDIA, equipped with a 64-bit quad-core ARM A57@1.43GHz CPU, 128-core NVIDIA Maxwell @921MHz GPU, and 4GB 64-bit LPDDR4 @1600MHz memory. The cloud server is equipped with a 64-bit 10-core Intel Xeon(R) W-2255 @3.70GHZ CPU, a GForce RTX 2080Ti GPU with 12GB memory, and 64GB of RAM. Communication between the cloud server and the edge devices uses a point-to-point WiFi connection, and network traffic simulation is controlled by the Linux Traffic Control tool, which can simulate network scenarios with different network bandwidth, communication quality, and latency.

\subsubsection{ Models and data sets}
To evaluate the performance of the cloud-edge collaborative security inference algorithm, the experiments use three chain topology DNN models: AlexNet, VGG16, and MobileNet v1, and the models are modified accordingly to be adapted and deployed in the CIS framework. Regarding the dataset, the CIFAR-10 \cite{krizhevsky2009learning} dataset is used for all model training and inference in our experiments to compare and validate the accuracy of our proposed methods. The dataset CIFAR-10 contains 50,000 RGB training images of size 32 × 32 and 10,000 test images for 10 classes.

\subsubsection{Baseline and Evaluation Metrics}
\textbf{(1) Inference latency performance}

(a) Cloud-only: The total inference latency generated by the DNN network processing only in the cloud.

(b) Device-only: The total inference delay generated by the DNN network for processing only at the edge devices.

(c) Neurosurgeon \cite{kang2017neurosurgeon}: The first proposed method for collaborative inference between edge devices and clouds based on DNN model partitioning, the total inference delay generated by DNN networks in collaboration at the cloud and edge.

\textbf{(2) Defensive performance of reconstruction attacks}

White-box Reconstruction Attack (WRA), based on the attack method provided by He et al \cite{he2020attacking}, an internal adversary from the cloud reconstructs the input data of the original sensitive edge device using the network structure and parameters of the shared model, as well as the uploaded intermediate output information.

Besides, we consider another Black-box Inverse-Network Attack (BINA) \cite{he2020attacking} from an external threat compared to the white-box reconstruction attack, where an adversary from the external does not have any a priori information about the network and parameters of the model and can only train an inverse model to reconstruct the input data from intermediate results through legitimate query requests using the Cloud Edge collaborative inference service.

In addition to the visualization of the reconstructed images to validate the proposed privacy-preserving mechanism, MSE, SSIM, and PSNR metrics, which typically measure the difference between the original and reconstructed images, are used to quantify the effectiveness of the reconstruction attack.

(a) Mean Squared Error (MSE): MSE measures the similarity between two images by calculating the cumulative squared error of the pixel values. the lower the MSE value, the higher the similarity between the two images ($A$ and $B$, image pixel size $m \cdot n$). The specific calculation is as follows:
\begin{equation}
MSE\left(A,B\right)=\frac{1}{m\cdot n}\sum_{i,j=1.1}^{m,n}\left \| A\left ( i,j\right )-B\left ( i,j\right )\right \|^2
\label{eq-11}
\end{equation}

(b) Structural similarity (SSIM): SSIM is a perception-based metric that measures the similarity between two images based on structural information. The specific calculation is as follows:
\begin{equation}
SSIM\left ( A,B\right )=\frac{\left ( 2\mu_A\mu_B+C_1\right )\left ( 2\sigma_{AB}+C_2\right )}{\left ( \mu_A^2+\mu_B^2+C_1\right )\left ( \sigma_A^2+\sigma_B^2+C_2\right )}
\label{eq-12}
\end{equation}
\noindent where $\mu_A$ and $\mu_B$ denote the mean values of pixels in images $A$ and $B$, respectively. $\sigma_A^2$ and $\sigma_B^2$ denote the variance, and $\sigma_{AB}$ denotes the covariance. In addition, $C_1$ and $C_2$ are constants and the value of SSIM is between $\left [0,1\right ]$, and higher SSIM values indicate higher similarity between the two images.

(c) Peak signal-to-noise ratio (PSNR): PSNR measures the similarity of two images by the peak error. the larger the PSNR value, the higher the image similarity. The specific calculation is as follows:
\begin{equation}
PSNR\left(A,B\right)=10\log_{10}\left(\frac{255^2}{MSE\left(A,B\right)} \right)
\label{eq-13}
\end{equation}

\subsection{Inference Latency Performance Analysis}
First, for the analysis of inference delay performance, our target networks are selected as AlexNet, VGG16, and MobileNet v1. The number of layers of the corresponding networks is given in Table \ref{table:1}. Meanwhile, in order to simulate the dynamically changing network quality between the real edge and the cloud, three networks with different conditions are simulated by the Linux Traffic Control tool, respectively, as shown in Table \ref{table:1} below.

\begin{table}[htb]
\centering

\tabcolsep=0.08cm
\caption{Simulation of networks with different communication rates and the number of layers corresponding to the target network.}
\begin{tabular}{@{}cclcccccc@{}}
\toprule
\textbf{Network conditions}         & \textbf{Poor} & \textbf{Medium}              & \textbf{Good} & \textbf{Excellence} & \textbf{Model Type} & \textbf{AlexNet} & \textbf{VGG-16} & \textbf{MobileNet v1} \\ \midrule
\textbf{Uplink transmission rate (Mpbs)} & 0.15       & \multicolumn{1}{c}{1.3} & 4          & 15         & \textbf{Number of Layers} & 13               & 34              & 54                    \\ \bottomrule
\end{tabular}
\label{table:1}
\end{table}

Subsequently, the proposed CIS cloud-side collaborative inference system achieves inference acceleration in comparison with several other approaches, including device-only, cloud-only, and Neurosurgeon, under different network quality conditions, given in Figure \ref{fig6}. It can be clearly found in Figure \ref{fig:5-6-(a)} that with poor network quality (uplink bandwidth as low as 0.15Mpbs), the transmission latency generated by sending raw data to the cloud is much larger than the computation latency in edge devices, and the inference speedup of cloud-only is significantly lower than that of device-only for DNN models of different sizes baseline. For the CIS and Neurosurgeon methods with the model partitioning mechanism, the inference speedup ratios for different models are still slightly higher than the baseline device-only method even with low network transmission quality, reaching 1.15x $\sim$1.65 for CIS and 0.88x $\sim$1.34x for Neurosurgeon.

With the improvement of network transmission quality, the bottleneck of computational performance on edge devices is gradually amplified, while the powerful computational power of cloud servers is released with the benefit of network quality improvement. As a result, the inference acceleration ratios of cloud-only, CIS, and Neurosurgeon are increased exponentially. When the network transmission rate reaches 4Mps, the inference speedup ratios of cloud-only reaches 7.12x, CIS reaches 13.56x and Neurosurgeon reaches 10.06x. Finally, when the network transmission rate reaches 15Mpbs, the ratio of transmission delay to total inference delay decreases further, and the inference speedup ratio of cloud-only is significantly better than other methods, reaching 17.89x for the MobileNet model. Moreover, Neurosurgeon's inference speedup ratio degrades compared to the network transmission rate of 4 Mbps. Nevertheless, CIS still slightly outperforms cloud-only in the inference speedup ratio for the AlexNet model.

In a word, with the continuous development of computing capability of edge devices, the computational schema of cloud-edge collaborative inference can make up for the computational latency bottleneck of cloud-only mode under the complex and variable network transmission quality. Meanwhile, the security and privacy of data is further enhanced by ensuring that the original data does not leave the device.

\begin{figure}[h]
	\centering
	\subfigure[\textbf{Poor network quality (0.15Mpbs)}]{
		\begin{minipage}[t]{0.45\textwidth}
			\includegraphics[width=1\textwidth]{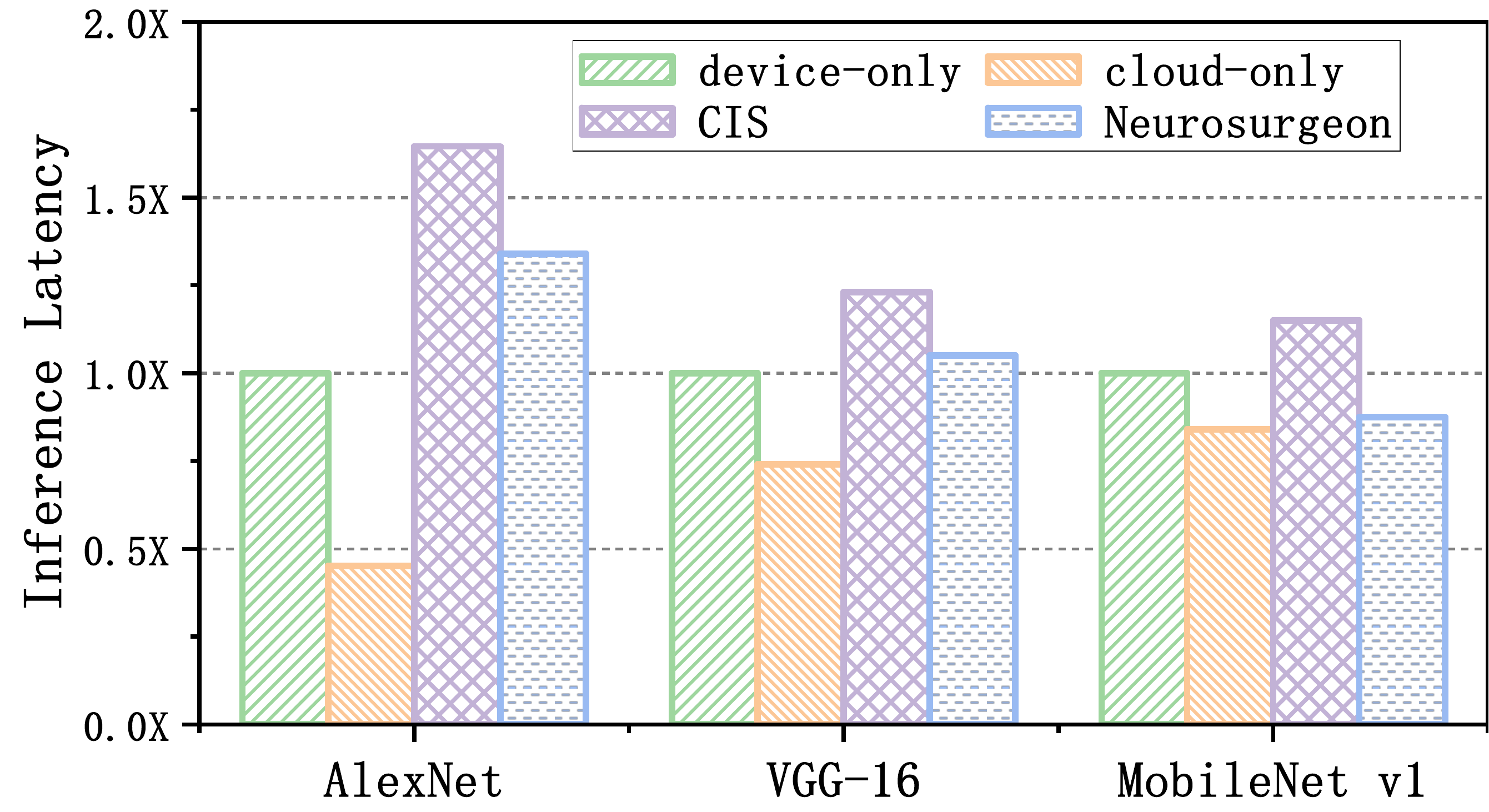}
		\end{minipage}
		\label{fig:5-6-(a)}
	}
    	\subfigure[\textbf{Medium network quality (1.3Mpbs)}]{
    		\begin{minipage}[t]{0.45\textwidth}
   		 	\includegraphics[width=1\textwidth]{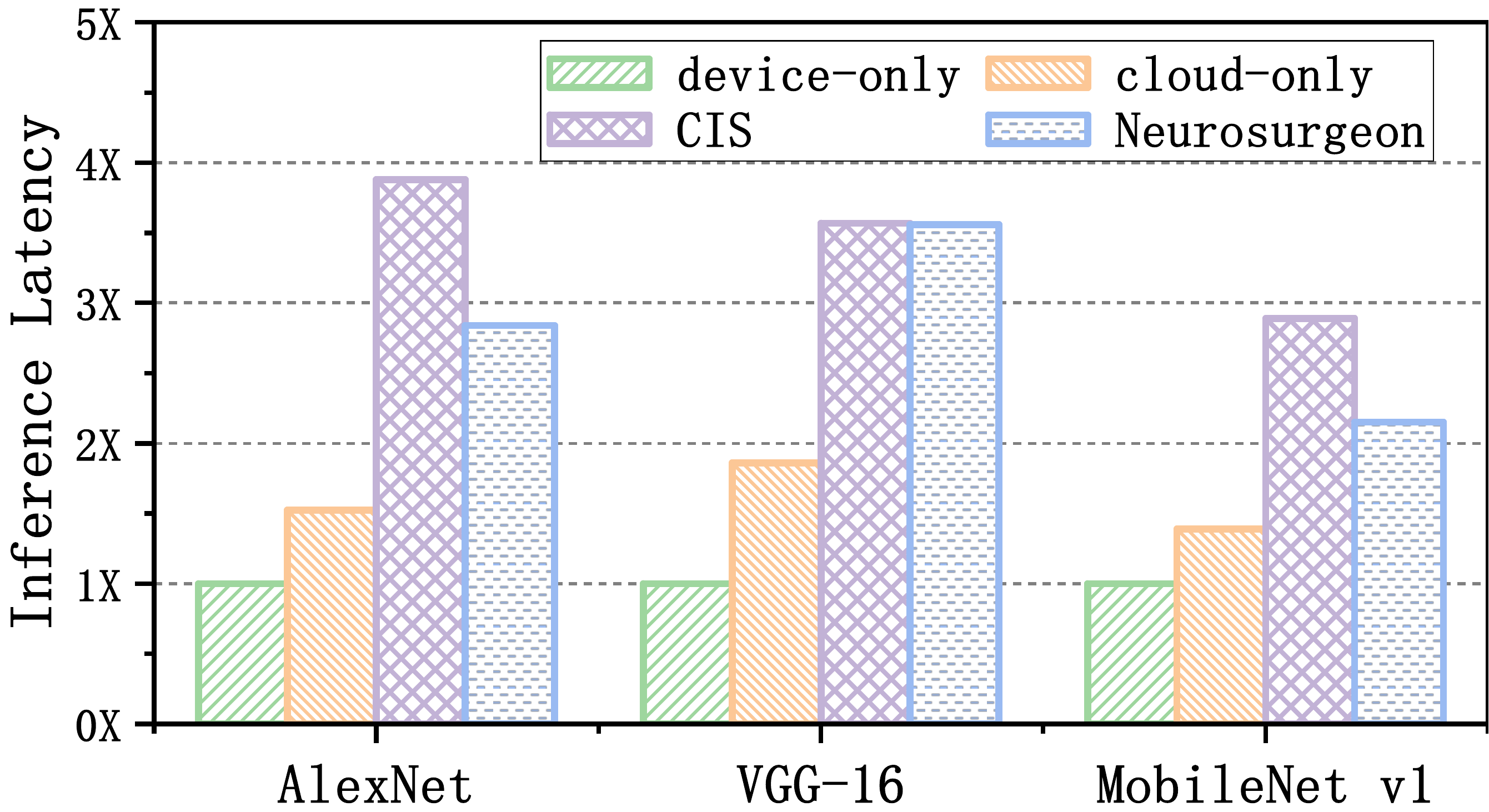}
    		\end{minipage}
		\label{fig:5-6-(b)}
    	}
    	\\
    	\subfigure[\textbf{Good network quality (4Mpbs)}]{
		\begin{minipage}[t]{0.45\textwidth}
			\includegraphics[width=1\textwidth]{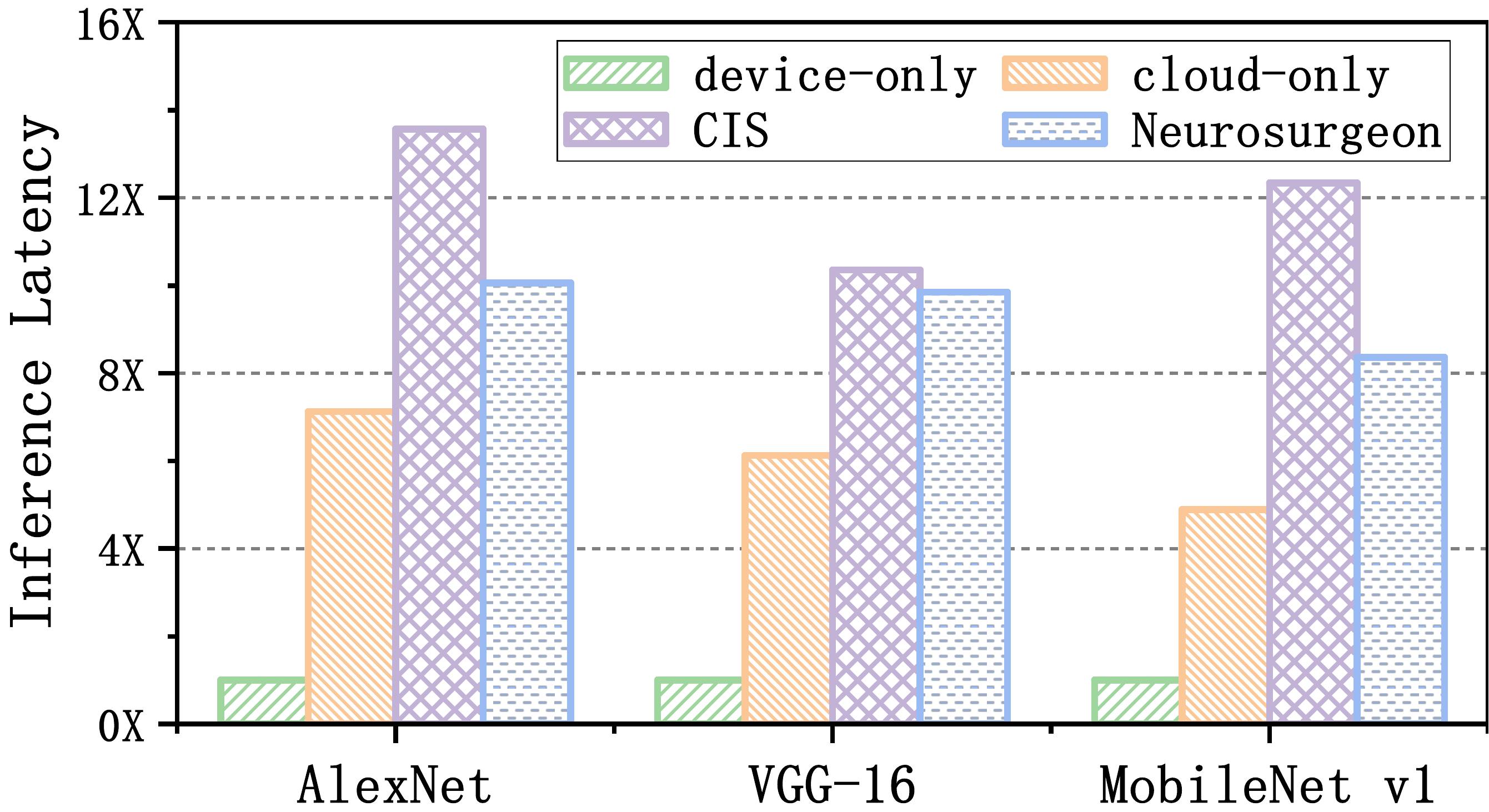}
		\end{minipage}
		\label{fig:5-6-(a)}
	}
    	\subfigure[\textbf{Excellence network quality (15Mpbs)}]{
    		\begin{minipage}[t]{0.45\textwidth}
   		 	\includegraphics[width=1\textwidth]{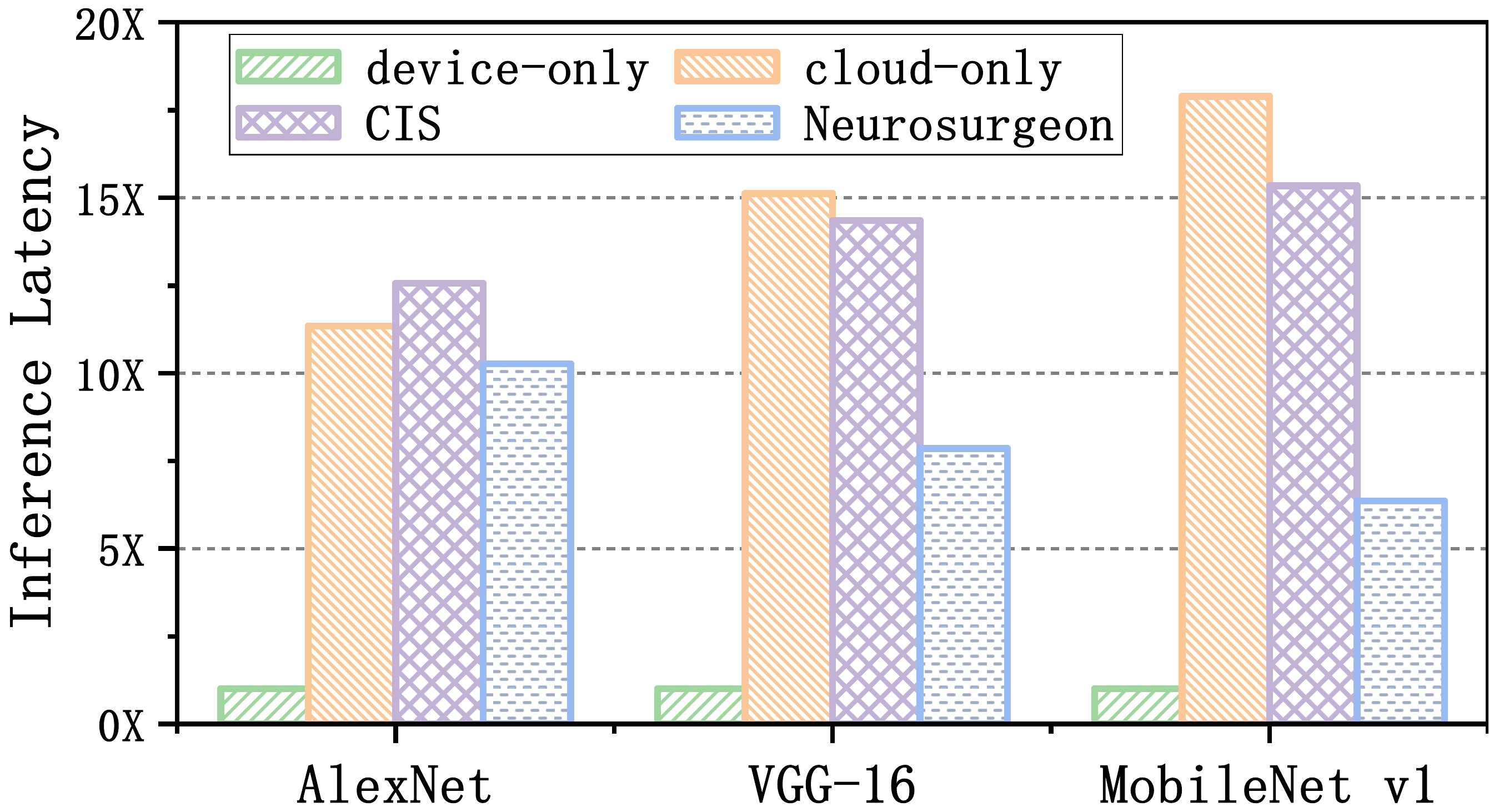}
    		\end{minipage}
		\label{fig:5-6-(d)}
    	}
	\caption{Comparison of inference latency acceleration achieved by CIS with different baseline methods under different network quality conditions}
	\label{fig6}
\end{figure}

\subsection{Defensive performance of white-box reconstruction attacks}
We utilize the regularized Maximum Likelihood Estimation (rMLE)-based white-box reconstruction attack (WRA) \cite{he2020attacking} to verify the defense performance of the proposed Collaborative-DP. Specifically, given the intermediate output results and the inference network and parameters $f^{device}_{\Theta_l}\left(x\right)$, the similarity of the reconstructed input $\hat{x}$ to the original input $x$ (indicating the posterior information observed by the adversary from the intermediate results) is measured by the Euclidean distance (ED), and the total variation ( TV) to represent the prior information of the original input, as follows:
\begin{align}
\label{eq-5-14}
\hat{x}&=\arg \min_x ED\left ( x,\hat{x}\right )+\lambda TV\left ( x\right )\\ \notag
&=\arg \min_x \left \|f^{device}_{\Theta_l}\left(x\right)-f^{device}_{\Theta_l}\left(\hat{x}\right) \right \|_2^2+\lambda \sum_{i,j}\left ( \left |\hat{x}_{i+1,j}-\hat{x}_{i,j} \right |^2+\left |\hat{x}_{i,j}-\hat{x}_{i,j+1} \right |^2\right )^{\beta/2}
\end{align}
\noindent where $i,j$ denote the pixel locations in the image, $\beta$ is the parameter that controls the image smoothness to avoid drastic changes inside the image, and $\lambda$ is the parameter that balances ED and TV.

Before further verifying the defense performance of the Collaborative-DP algorithm in CIS against WRA attacks, we give a visualization of the privacy budget allocation based on the rank of feature submaps (as in Eq. \ref{eq-9}) in Collaborative-DP. As shown in Fig. \ref{fig7}, the first row shows the original CIFA-10 input images, and the feature submaps output by the collaborative inference network (VGG-16) at the partition layer sorted by their rank, in that order. Given the total privacy budget $\epsilon$, Collaborative-DP allocates the budget reasonably according to the rank order, and the lower the feature submaps of lower rank contain lower amount of available information, the privacy budget allocated for them is correspondingly lower, and the added noise becomes higher (X-axis direction in Figure \ref{fig7}). As the total privacy budget $\epsilon$ keeps decreasing (in the Y-axis direction in Fig. \ref{fig7}), the noise added to the corresponding individual feature submaps also increases significantly, corresponding to the increasing intensity of privacy. Even so Collaborative-DP's rank-based budget allocation mechanism still ensures that feature submaps with high rank have relatively high availability.

 \begin{figure}
 \centering
  \begin{subfigure}
  \centering
  \includegraphics[width=0.7\linewidth]{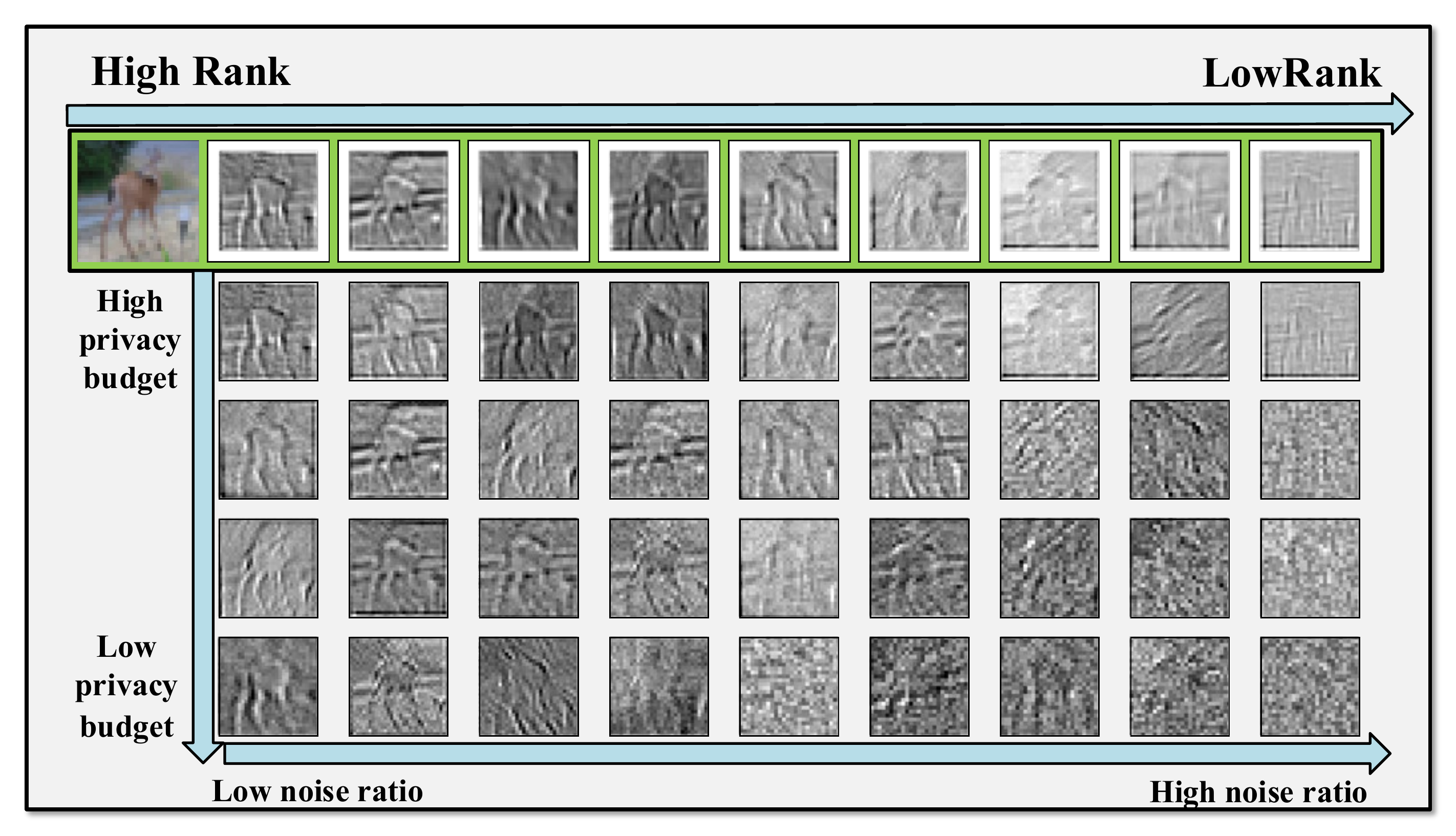}
  \caption{The effect of privacy budget allocation based on feature submaps rank.}
  \label{fig7}
  \end{subfigure}
\end{figure}
Next, the defense effect of Collaborative-DP algorithm in CIS against WRA attack with different privacy protection strengths is given in Fig. \ref{fig8}, where the target model is VGG-16 as well as CIFAR-10 as the target dataset. Analyzed from a visual perspective, the WRA attack is very effective in reconstructing the input image on the edge device almost completely. As the given total privacy budget $\epsilon$ decreases gradually from high, the noise addition of Collaborative-DP to the intermediate output increases and the attack effect of WRA decreases gradually. When $\epsilon \in \left[10,30\right]$, the content of the target data can no longer be clearly distinguished from the visual point of view.
\begin{figure}
 \centering
  \begin{subfigure}
  \centering
  \includegraphics[width=0.7\linewidth]{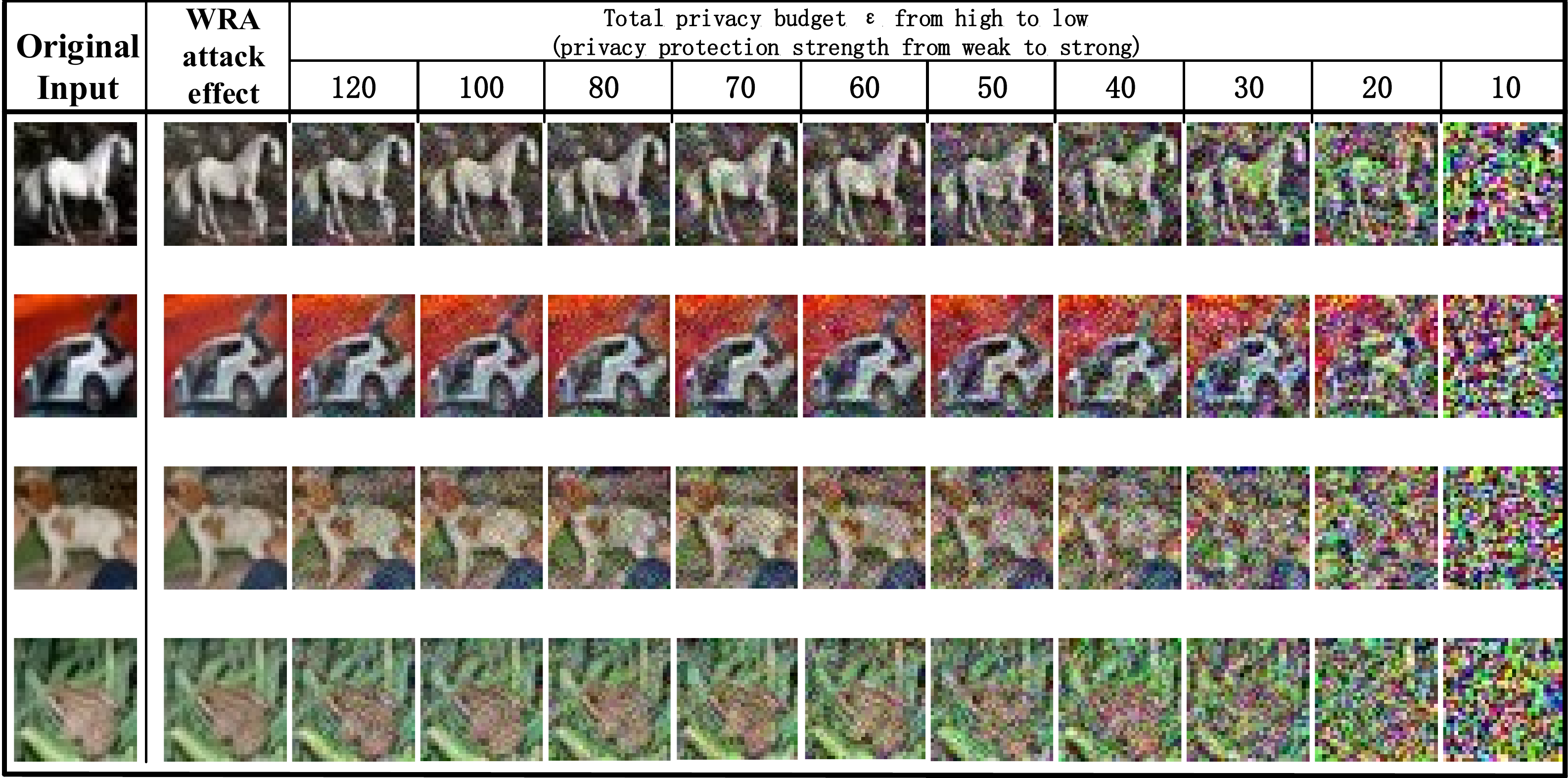}
  \caption{Defense effectiveness of Collaborative-DP against WRA attacks with different privacy protection strengths.}
  \label{fig8}
  \end{subfigure}
\end{figure}

In addition to visualizing the reconstructed images of WRA attacks, MSE, SSIM and PSNR metrics are used in Fig. \ref{fig9} to quantify the effectiveness of defense against WRA attacks. Also, Collaborative-DP is compared with Non-DP without added noise (as the baseline algorithm) and the plain noise addition strategy \cite{ryu2022can} proposed by Ryu et al. (which we named Native-DP for comparison purposes), respectively. For the MSE metric, a lower value means the closer the reconstructed image is to the original image, and the more effective the reconstructed attack is. When $10 < \epsilon < 50$, the MSE value of Native-DP is much lower than that of Collaborative-DP, implying that Collaborative-DP can achieve better defense against reconstruction attacks with a smaller privacy budget. Contrary to the MSE metric, a larger PSNR value indicates that the reconfiguration attack generates a higher quality image and a more effective attack. It can be observed very directly in Fig. \ref{fig9} that the PSNR value of Collaborative-DP is smaller than that of Native-DP in all privacy budget ranges, which also indicates that Collaborative-DP has higher defense performance against WRA attacks with the same privacy budget. Finally, the SSIM metric measures the similarity between the reconstructed image and the original image by structural similarity, and higher values indicate that the reconstructed attack is more effective. In Fig. \ref{fig9}, it can be found that Native-DP is superior to Collaborative-DP algorithm under SSIM metric. A reasonable explanation is that since the SSIM metric reflects the properties of the object structure in the scene from the perspective of image composition, while Collaborative-DP differs from Native-DP by assigning privacy budget based on the size of the feature subgraph rank. The structural properties of some feature subgraphs are preserved to some extent while trade-offs are made between privacy strength and availability.

 \begin{figure}
 \centering
  \begin{subfigure}
  \centering
  \includegraphics[width=0.7\linewidth]{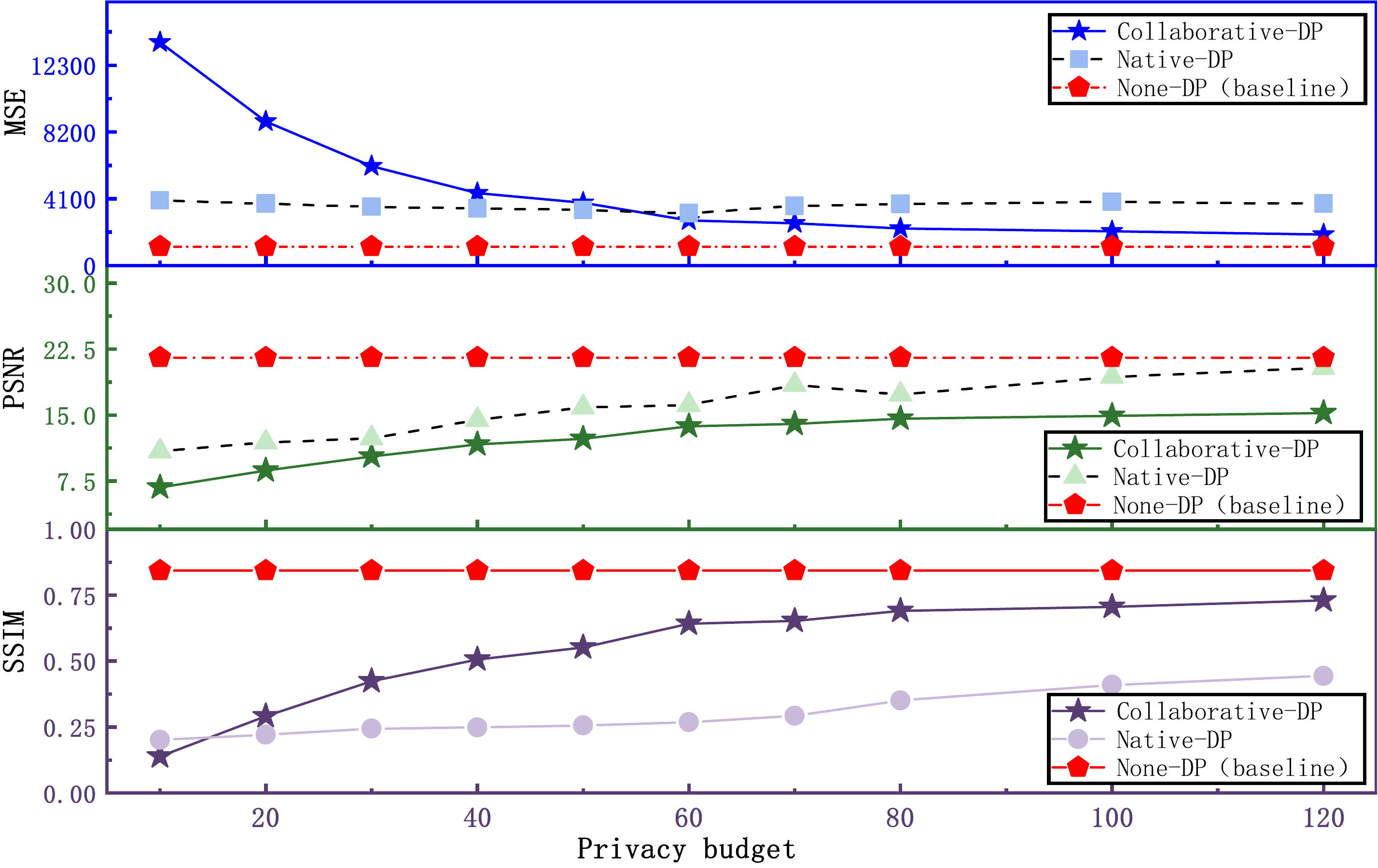}
  \caption{Quantitative defense effects of Collaborative-DP and different baseline methods against white-box reconstruction attacks.}
  \label{fig9}
  \end{subfigure}
\end{figure}

\subsection{Defensive performance of Black-box Inverse-Network Attack}
In this subsection, compared to the white-box reconstruction attack, we consider another black-box attack from an external threat with more restrictive conditions. Specifically, an adversary from external does not have any a priori information about the network and parameters $f_{\Theta}$ of the model, and can only use the cloud-edge collaborative inference service through legitimate query requests. Suppose an external adversary can get the intermediate results computed from the edge device with arbitrary input $x$, and then reconstruct the input data from the intermediate results by training an inverse model $f^{-1}_{\Theta_l}$, as follows:
\begin{equation}
\label{eq-5-14}
f^{-1}_{\Theta_l}=\arg \min_g \frac{1}{m}\sum_{i=1}^{m}\left \| g\left ( f_{\Theta_l}\left ( x_i\right )\right )-x_i\right \|^2
\end{equation}
\noindent where $\left\{x_i\right\}^m_{i=1}$ is the training set generated for the inverse model $g$, and the output $\left\{f_{\Theta_l}\left ( x_i\right ),x_i\right\}$ obtained by legitimate request is used as the sample to train the inverse model $g$ and the parameters $f^{- 1}_{\Theta_l}$. An external adversary can then reconstruct the sensitive data $\hat{x}=f^{-1}_{\Theta_l}\left(v_m\right)$ by the inverse model and the obtained intermediate result $v_m$.

Similarly, the defense effectiveness of Collaborative-DP algorithm in CIS against BINA attack under different privacy protection strengths is given in Fig. \ref{fig10}, where the target model is VGG-16 as well as CIFAR-10 as the target dataset. First, analyzing from the visual perspective, although the data reconstructed by BINA can still be distinguished by the naked eye, the attack is less effective than the white-box attack model of WRA. This stems from the restriction of the attack condition that the external adversary lacks prior knowledge of the model and strongly relies on the generated training set $\left\{f_{\Theta_l}\left ( x_i\right ),x_i\right\}$ to fit the constructed inverse model parameters $f^{-1}_{\Theta_l}$.

However, both Collaborative-DP and Native-DP methods add noise ($\hat{v}'=f_{\Theta_l}\left ( x\right )+noise$) to the intermediate results (also the training set generated by the inverse network), and the overfitted inverse model is very sensitive to noise. Therefore, it can be found in Figure \ref{fig10} that when the privacy budget $\epsilon < 500$, the Collaborative-DP algorithm has been able to resist the BINA attack very well. The different quantitative metrics given in Figs. \ref{fig:5-11} also corroborate the above analysis. Under the condition of $\epsilon = 100$ the Collaborative-DP algorithm makes the MSE metrics fast approaching with 20,000 while the MSE and SSIM evaluation metrics are also much lower than the baseline method of Non-DP without adding noise. Based on the above analysis, the defense performance of Collaborative-DP algorithm against BINA attack is already excellent when the privacy budget $\epsilon < 500$.

 \begin{figure}
 \centering
  \begin{subfigure}
  \centering
  \includegraphics[width=0.7\linewidth]{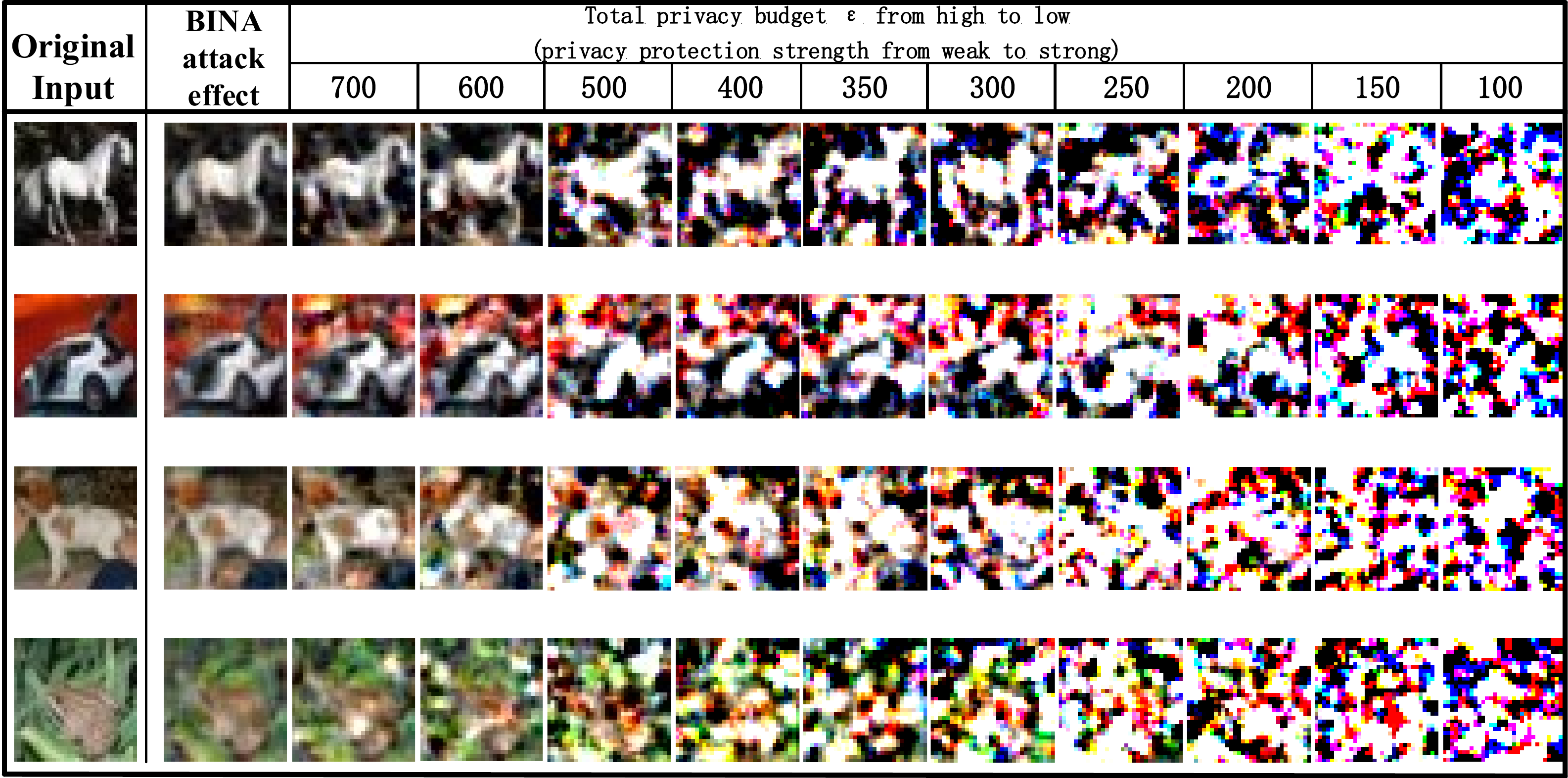}
  \caption{Effectiveness of Collaborative-DP's defense against BINA attacks with different privacy protection strengths.}
  \label{fig10}
  \end{subfigure}
\end{figure}

 \begin{figure}
 \centering
  \begin{subfigure}
  \centering
  \includegraphics[width=0.7\linewidth]{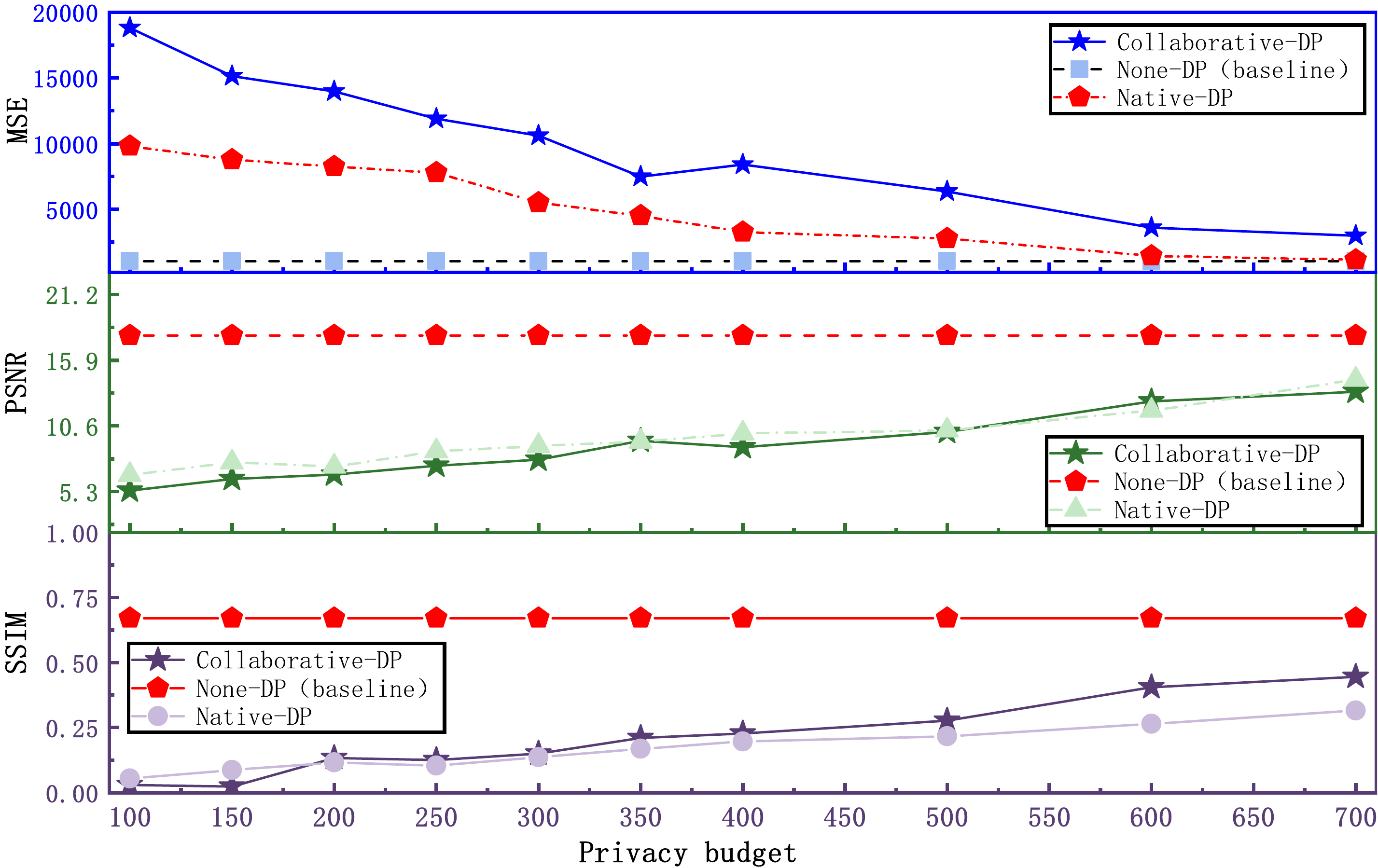}
  \caption{Quantitative Defense Effectiveness of Collaborative-DP and Different Baseline Methods against Black Box Inverse Model Attacks.}
  \label{fig11}
  \end{subfigure}
\end{figure}

\subsection{Analysis of inference accuracy}
The above analysis of the security defense performance of Collaborative-DP facing different attack scenarios is presented. However, blindly achieving high strength privacy protection capability by reducing the privacy budget can heavily sacrifice the accuracy of model inference. Therefore, in this subsection, experiments are conducted to analyze how much impact on the inference accuracy of Collaborative-DP can be caused by different privacy budgets, and how to pick the appropriate privacy budget to achieve the trade-off between privacy and usability.

As can be seen in the following figure \ref{fig12}, Collaborative-DP has a much lower impact on the accuracy of the original network than the Native-DP approach through the refined privacy budget allocation mechanism. When $\epsilon>10$, the prediction accuracy of the model under Collaborative-DP protection has reached 82.64\%, which is slightly lower than the 86.69\% of Native-DP and much higher than the 56.56\% of Native-DP. And the Native-DP mechanism can reach the best balance of privacy and usability only when $\epsilon>30$. As the previous analysis of the defense performance of the attacks, the minimum privacy budgets to effectively defend against WRA attacks and BINA attacks are 30 and 500, respectively, which are much higher than the privacy budgets required for the availability equilibrium point of Collaborative-DP and Native-DP. Thus, the mechanisms of both Collaborative-DP and Native-DP can protect against advanced threats from internal and external while guaranteeing the accuracy of model prediction, and Collaborative-DP requires less privacy budget. Moreover, after $\epsilon>10$ of Collaborative-DP, the prediction accuracy of the network has stabilized and cannot overlap with the accuracy of the original Non-DP no matter how much the privacy budget is scaled up. A reasonable explanation is that a certain degree of accuracy loss is caused by the fact that Collaborative-DP uses a fixed threshold $\boldsymbol{C_m}$ for feature mapping to clipping in order to estimate the global sensitivity.

 \begin{figure}
 \centering
  \begin{subfigure}
  \centering
  \includegraphics[width=0.7\linewidth]{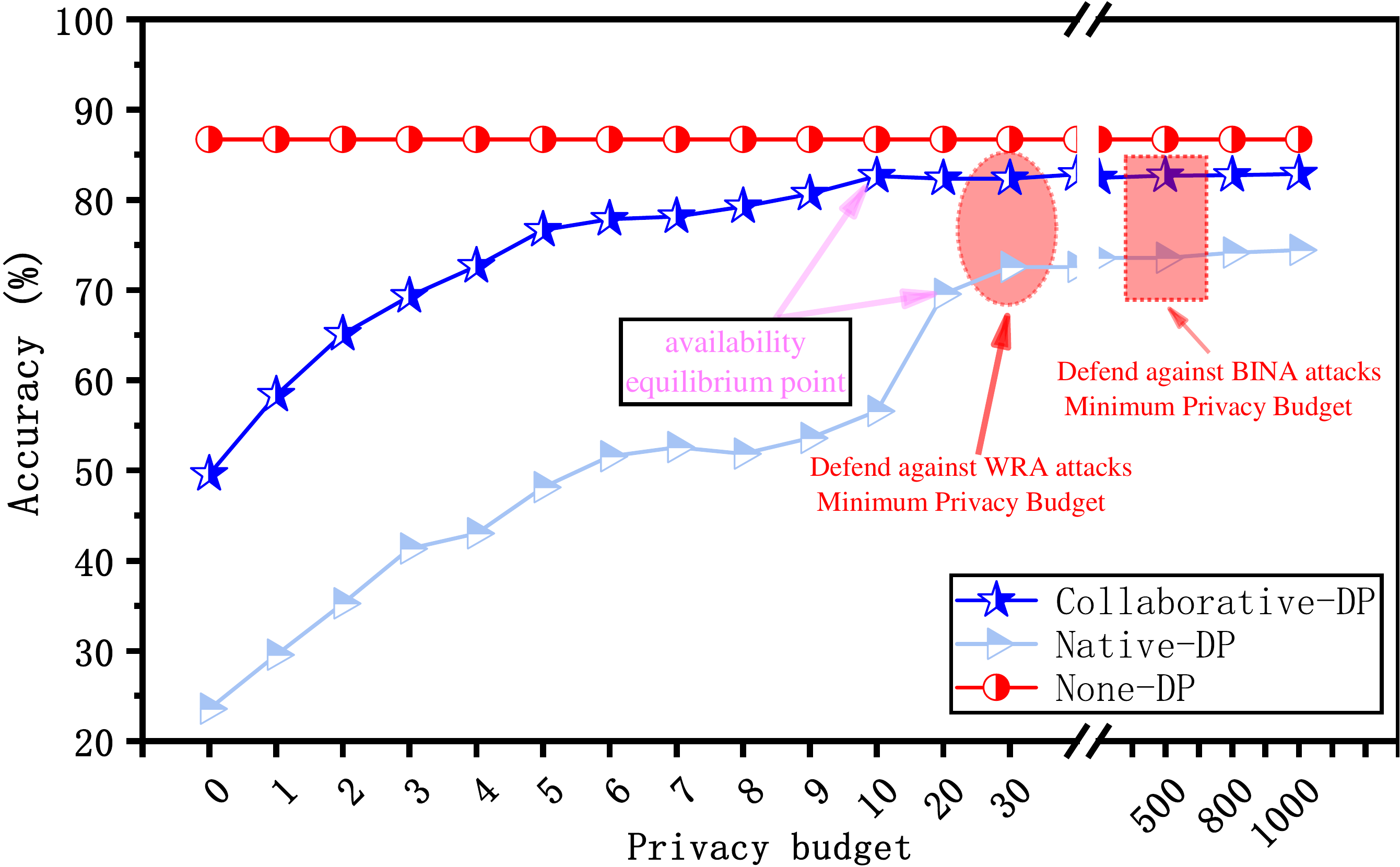}
  \caption{Quantitative Defense Effectiveness of Collaborative-DP and Different Baseline Methods against Black Box Inverse Model Attacks.}
  \label{fig12}
  \end{subfigure}
\end{figure}

\section{Conclusion}
In this paper, a secure privacy inference framework for cloud-edge collaboration is proposed, which supports adaptively partitioning the network according to the dynamically changing network bandwidth and fully releases the computational power of edge devices. Meanwhile, the partition point is selected with full consideration of the amount of information of intermediate results that need to be uploaded, and refined noise is added to them to achieve a differential privacy protection mechanism. Finally, a realistic cloud-edge collaborative inference computing scenario is constructed to evaluate the effectiveness of inference latency and model partitioning on resource-constrained edge devices. Meanwhile, state-of-the-art cloud-edge collaborative inference reconstruction attacks are employed to evaluate the practical usability of the end-to-end privacy-preserving mechanism of CIS.

\section*{Acknowledgments}
This was was supported in part by......

\bibliographystyle{unsrt}  
\bibliography{references}

\begin{thebibliography}{10}

\bibitem{mao2017survey}
Yuyi Mao, Changsheng You, Jun Zhang, Kaibin Huang, and Khaled~B Letaief.
\newblock A survey on mobile edge computing: The communication perspective.
\newblock {\em IEEE communications surveys \& tutorials}, 19(4):2322--2358,
  2017.

\bibitem{wang2019ecass}
Xiong Wang, Tianpeng Wei, Linghe Kong, Liang He, Fan Wu, and Guihai Chen.
\newblock Ecass: Edge computing based auxiliary sensing system for self-driving
  vehicles.
\newblock {\em Journal of Systems Architecture}, 97:258--268, 2019.

\bibitem{kuang2021cooperative}
Zhufang Kuang, Zhihao Ma, Zhe Li, and Xiaoheng Deng.
\newblock Cooperative computation offloading and resource allocation for delay
  minimization in mobile edge computing.
\newblock {\em Journal of Systems Architecture}, 118:102167, 2021.

\bibitem{siriwardhana2021survey}
Yushan Siriwardhana, Pawani Porambage, Madhusanka Liyanage, and Mika
  Ylianttila.
\newblock A survey on mobile augmented reality with 5g mobile edge computing:
  architectures, applications, and technical aspects.
\newblock {\em IEEE Communications Surveys \& Tutorials}, 23(2):1160--1192,
  2021.

\bibitem{kang2017neurosurgeon}
Yiping Kang, Johann Hauswald, Cao Gao, Austin Rovinski, Trevor Mudge, Jason
  Mars, and Lingjia Tang.
\newblock Neurosurgeon: Collaborative intelligence between the cloud and mobile
  edge.
\newblock {\em ACM SIGARCH Computer Architecture News}, 45(1):615--629, 2017.

\bibitem{li2019edge}
En~Li, Liekang Zeng, Zhi Zhou, and Xu~Chen.
\newblock Edge ai: On-demand accelerating deep neural network inference via
  edge computing.
\newblock {\em IEEE Transactions on Wireless Communications}, 19(1):447--457,
  2019.

\bibitem{zhang2021autodidactic}
Letian Zhang, Lixing Chen, and Jie Xu.
\newblock Autodidactic neurosurgeon: Collaborative deep inference for mobile
  edge intelligence via online learning.
\newblock In {\em Proceedings of the Web Conference 2021}, pages 3111--3123,
  2021.

\bibitem{hu2019dynamic}
Chuang Hu, Wei Bao, Dan Wang, and Fengming Liu.
\newblock Dynamic adaptive dnn surgery for inference acceleration on the edge.
\newblock In {\em IEEE INFOCOM 2019-IEEE Conference on Computer
  Communications}, pages 1423--1431. IEEE, 2019.

\bibitem{he2020attacking}
Zecheng He, Tianwei Zhang, and Ruby~B Lee.
\newblock Attacking and protecting data privacy in edge--cloud collaborative
  inference systems.
\newblock {\em IEEE Internet of Things Journal}, 8(12):9706--9716, 2020.

\bibitem{yeom2018privacy}
Samuel Yeom, Irene Giacomelli, Matt Fredrikson, and Somesh Jha.
\newblock Privacy risk in machine learning: Analyzing the connection to
  overfitting.
\newblock In {\em 2018 IEEE 31st computer security foundations symposium
  (CSF)}, pages 268--282. IEEE, 2018.

\bibitem{shokri2017membership}
Reza Shokri, Marco Stronati, Congzheng Song, and Vitaly Shmatikov.
\newblock Membership inference attacks against machine learning models.
\newblock In {\em 2017 IEEE symposium on security and privacy (SP)}, pages
  3--18. IEEE, 2017.

\bibitem{liu2017oblivious}
Jian Liu, Mika Juuti, Yao Lu, and Nadarajah Asokan.
\newblock Oblivious neural network predictions via minionn transformations.
\newblock In {\em Proceedings of the 2017 ACM SIGSAC conference on computer and
  communications security}, pages 619--631, 2017.

\bibitem{gilad2016cryptonets}
Ran Gilad-Bachrach, Nathan Dowlin, Kim Laine, Kristin Lauter, Michael Naehrig,
  and John Wernsing.
\newblock Cryptonets: Applying neural networks to encrypted data with high
  throughput and accuracy.
\newblock In {\em International conference on machine learning}, pages
  201--210. PMLR, 2016.

\bibitem{mireshghallah2020principled}
Fatemehsadat Mireshghallah, Mohammadkazem Taram, Ali Jalali, Ahmed~Taha
  Elthakeb, Dean Tullsen, and Hadi Esmaeilzadeh.
\newblock A principled approach to learning stochastic representations for
  privacy in deep neural inference.
\newblock {\em arXiv preprint arXiv:2003.12154}, 2020.

\bibitem{wang2018not}
Ji~Wang, Jianguo Zhang, Weidong Bao, Xiaomin Zhu, Bokai Cao, and Philip~S Yu.
\newblock Not just privacy: Improving performance of private deep learning in
  mobile cloud.
\newblock In {\em Proceedings of the 24th ACM SIGKDD international conference
  on knowledge discovery \& data mining}, pages 2407--2416, 2018.

\bibitem{mireshghallah2020shredder}
Fatemehsadat Mireshghallah, Mohammadkazem Taram, Prakash Ramrakhyani, Ali
  Jalali, Dean Tullsen, and Hadi Esmaeilzadeh.
\newblock Shredder: Learning noise distributions to protect inference privacy.
\newblock In {\em Proceedings of the Twenty-Fifth International Conference on
  Architectural Support for Programming Languages and Operating Systems}, pages
  3--18, 2020.

\bibitem{dwork2006calibrating}
Cynthia Dwork, Frank McSherry, Kobbi Nissim, and Adam Smith.
\newblock Calibrating noise to sensitivity in private data analysis.
\newblock In {\em Theory of cryptography conference}, pages 265--284. Springer,
  2006.

\bibitem{zhao2022survey}
Ying Zhao and Jinjun Chen.
\newblock A survey on differential privacy for unstructured data content.
\newblock {\em ACM Computing Surveys (CSUR)}, 54(10s):1--28, 2022.

\bibitem{dwork2014algorithmic}
Cynthia Dwork, Aaron Roth, et~al.
\newblock The algorithmic foundations of differential privacy.
\newblock {\em Foundations and Trends{\textregistered} in Theoretical Computer
  Science}, 9(3--4):211--407, 2014.

\bibitem{mcsherry2009privacy}
Frank~D McSherry.
\newblock Privacy integrated queries: an extensible platform for
  privacy-preserving data analysis.
\newblock In {\em Proceedings of the 2009 ACM SIGMOD International Conference
  on Management of data}, pages 19--30, 2009.

\bibitem{mach2017mobile}
Pavel Mach and Zdenek Becvar.
\newblock Mobile edge computing: A survey on architecture and computation
  offloading.
\newblock {\em IEEE Communications Surveys \& Tutorials}, 19(3):1628--1656,
  2017.

\bibitem{teerapittayanon2017distributed}
Surat Teerapittayanon, Bradley McDanel, and Hsiang-Tsung Kung.
\newblock Distributed deep neural networks over the cloud, the edge and end
  devices.
\newblock In {\em 2017 IEEE 37th international conference on distributed
  computing systems (ICDCS)}, pages 328--339. IEEE, 2017.

\bibitem{ko2018edge}
Jong~Hwan Ko, Taesik Na, Mohammad~Faisal Amir, and Saibal Mukhopadhyay.
\newblock Edge-host partitioning of deep neural networks with feature space
  encoding for resource-constrained internet-of-things platforms.
\newblock In {\em 2018 15th IEEE International Conference on Advanced Video and
  Signal Based Surveillance (AVSS)}, pages 1--6. IEEE, 2018.

\bibitem{zhang2021deepslicing}
Shuai Zhang, Sheng Zhang, Zhuzhong Qian, Jie Wu, Yibo Jin, and Sanglu Lu.
\newblock Deepslicing: collaborative and adaptive cnn inference with low
  latency.
\newblock {\em IEEE Transactions on Parallel and Distributed Systems},
  32(9):2175--2187, 2021.

\bibitem{banitalebi2021auto}
Amin Banitalebi-Dehkordi, Naveen Vedula, Jian Pei, Fei Xia, Lanjun Wang, and
  Yong Zhang.
\newblock Auto-split: a general framework of collaborative edge-cloud ai.
\newblock In {\em Proceedings of the 27th ACM SIGKDD Conference on Knowledge
  Discovery \& Data Mining}, pages 2543--2553, 2021.

\bibitem{manasi2020neupart}
Susmita~Dey Manasi, Farhana~Sharmin Snigdha, and Sachin~S Sapatnekar.
\newblock Neupart: Using analytical models to drive energy-efficient
  partitioning of cnn computations on cloud-connected mobile clients.
\newblock {\em IEEE Transactions on Very Large Scale Integration (VLSI)
  Systems}, 28(8):1844--1857, 2020.

\bibitem{xu2020energy}
Zichuan Xu, Liqian Zhao, Weifa Liang, Omer~F Rana, Pan Zhou, Qiufen Xia,
  Wenzheng Xu, and Guowei Wu.
\newblock Energy-aware inference offloading for dnn-driven applications in
  mobile edge clouds.
\newblock {\em IEEE Transactions on Parallel and Distributed Systems},
  32(4):799--814, 2020.

\bibitem{he2019model}
Zecheng He, Tianwei Zhang, and Ruby~B Lee.
\newblock Model inversion attacks against collaborative inference.
\newblock In {\em Proceedings of the 35th Annual Computer Security Applications
  Conference}, pages 148--162, 2019.

\bibitem{wang2020differential}
Jinyan Wang, Zhou Tan, Xianxian Li, Yuhang Hu, et~al.
\newblock Differential privacy preservation in interpretable
  feedforward-designed convolutional neural networks.
\newblock In {\em 2020 IEEE 19th International Conference on Trust, Security
  and Privacy in Computing and Communications (TrustCom)}, pages 631--638.
  IEEE, 2020.

\bibitem{mao2018learning}
Yunlong Mao, Shanhe Yi, Qun Li, Jinghao Feng, Fengyuan Xu, and Sheng Zhong.
\newblock Learning from differentially private neural activations with edge
  computing.
\newblock In {\em 2018 IEEE/ACM Symposium on Edge Computing (SEC)}, pages
  90--102. IEEE, 2018.

\bibitem{goldsmith2005wireless}
Andrea Goldsmith.
\newblock {\em Wireless communications}.
\newblock Cambridge university press, 2005.

\bibitem{9155237}
Thaha Mohammed, Carlee Joe-Wong, Rohit Babbar, and Mario~Di Francesco.
\newblock Distributed inference acceleration with adaptive dnn partitioning and
  offloading.
\newblock In {\em IEEE INFOCOM 2020 - IEEE Conference on Computer
  Communications}, pages 854--863, 2020.

\bibitem{abadi2016deep}
Martin Abadi, Andy Chu, Ian Goodfellow, H~Brendan McMahan, Ilya Mironov, Kunal
  Talwar, and Li~Zhang.
\newblock Deep learning with differential privacy.
\newblock In {\em Proceedings of the 2016 ACM SIGSAC conference on computer and
  communications security}, pages 308--318, 2016.

\bibitem{lin2020hrank}
Mingbao Lin, Rongrong Ji, Yan Wang, Yichen Zhang, Baochang Zhang, Yonghong
  Tian, and Ling Shao.
\newblock Hrank: Filter pruning using high-rank feature map.
\newblock In {\em Proceedings of the IEEE/CVF conference on computer vision and
  pattern recognition}, pages 1529--1538, 2020.

\bibitem{krizhevsky2009learning}
Alex Krizhevsky, Geoffrey Hinton, et~al.
\newblock Learning multiple layers of features from tiny images.
\newblock 2009.

\bibitem{ryu2022can}
Jihyeon Ryu, Yifeng Zheng, Yansong Gao, Alsharif Abuadbba, Junyaup Kim, Dongho
  Won, Surya Nepal, Hyoungshick Kim, and Cong Wang.
\newblock Can differential privacy practically protect collaborative deep
  learning inference for iot?
\newblock {\em Wireless Networks}, pages 1--21, 2022.

\end{thebibliography}

\end{document}